\documentclass[12pt,draftcls,onecolumn]{IEEEtran}

\usepackage{amsmath,amssymb}
\usepackage{subfigure}
\usepackage{graphicx,graphics,color,psfrag}
\usepackage{cite,balance}
\usepackage{caption}
\allowdisplaybreaks
\usepackage{algorithm}
\usepackage{accents}
\usepackage{amsthm}
\usepackage{bm}
\usepackage{algorithmic}
\usepackage[english]{babel}
\usepackage{multirow}
\usepackage{enumerate}
\usepackage{cases}
\usepackage{stfloats}
\usepackage{dsfont}
\usepackage{color,soul}
\usepackage{amsfonts}
\usepackage{cite,graphicx,amsmath,amssymb}
\usepackage{subfigure}
\usepackage{fancyhdr}
\usepackage{hhline}
\usepackage{graphicx,graphics}
\usepackage{array,color}
\usepackage{booktabs}
\usepackage{diagbox}
\usepackage{algorithm}
\usepackage{algorithmic}
\usepackage{indentfirst}

\usepackage{amsmath}

\usepackage{lipsum}
\usepackage{mwe}
\newtheorem{prop}{Proposition}
\newtheorem{Problem}{Problem}

\newcommand{\tabincell}[2]{\begin{tabular}{@{}#1@{}}#2\end{tabular}}


\newcounter{parentnumber}

\theoremstyle{definition}
\newtheorem{defn}{Definition}

\newtheorem{definition}{Definition}

\newtheorem{lemma}{Lemma}

\newtheorem{remark}{\bf Remark}
\def\phi{\varphi}

\def\({\left(}
\def\){\right)}

\def\b0{{\mathbf{0}}}

\begin{document}
	
\title{\Huge Cooperative Multi-Point Vehicular Positioning Using Millimeter-Wave  Surface Reflection}
	
\author{Zezhong Zhang, Seung-Woo Ko, Rui Wang, and Kaibin Huang
\thanks{Zezhong Zhang and Kaibin Huang are with The  University of  Hong Kong, Hong Kong (Email: \{zzzhang, huangkb\}@eee.hku.hk). Seung-Woo Ko is with Korea Maritime and Ocean University, Korea (email: swko@kmou.ac.kr). Rui Wang is with  Southern University of Science and Technology, China (Email: wang.r@sustech.edu.cn). Part of this work was presented in IEEE GLOBECOM 2019~\cite{Globecom}.}}
	
\maketitle
	
\begin{abstract}	
Multi-point vehicular positioning is one essential operation for autonomous vehicles. However, the state-of-the-art positioning technologies, relying on reflected signals from a target (i.e., RADAR and LIDAR), cannot work without \emph{line-of-sight} (LoS). Besides, it takes significant time for environment scanning and object recognition with potential detection inaccuracy,
especially in complex urban situations. Some recent fatal accidents involving autonomous vehicles further expose such limitations.  In this paper, we aim at overcoming these limitations by proposing a novel {relative positioning} approach, called \emph{Cooperative Multi-point Positioning} (COMPOP). The COMPOP establishes cooperation between a \emph{target vehicle} (TV) and a \emph{sensing vehicle} (SV) if a LoS path exists, where a TV explicitly lets an SV to know the TV's existence by transmitting positioning waveforms. This cooperation makes it possible to remove the time-consuming scanning and target recognizing processes, facilitating real-time positioning. One prerequisite for the cooperation is a clock synchronization between a pair of TV and SV. To this end, we use a \emph{phase-differential-of-arrival} (PDoA) based approach to remove the TV-SV clock difference from the received signal. With clock difference correction,  the TV's position can be obtained via peak detection over a 3D power spectrum constructed by a \emph{Fourier transform} (FT) based algorithm. {The COMPOP also incorporates nearby vehicles, without knowing their locations, into the above cooperation for the case without a LoS path.} Specifically, several strong \emph{non-LoS} (NLoS) links from the TV to the SV can be generated via mirror-like reflections over the neighboring vehicles' metal surfaces. Following procedures in the LoS case, virtual TVs mirrored by nearby vehicles can be detected. By exploiting the geometric relation between the virtual and actual TVs, COMPOP can be achieved by intelligently combining the virtual TVs to position the actual TV. The effectiveness of the COMPOP is verified by several simulations concerning practical channel parameters.
\end{abstract}

\section{Introduction}

Vehicular positioning is one of the most important operations for autonomous driving and a challenging one as it requires high accuracy and low latency~\cite{Husheng,Magzine}. Presently, for long-range positioning (e.g., hundreds of meters), a \emph{target vehicle} (TV) is abstracted as a single point, and its GPS location is communicated to the \emph{sensing vehicle} (SV) over either a vehicle-to-vehicle link or across a wireless network~\cite{PPP, MobileCentric}. For medium to short-range positioning (tens to several meters), the single-point abstraction of the TV no longer suffices, and its geometric information (e.g., size and orientation) is also required for safe and accurate driving. In these ranges, the popular positioning technologies include \emph{RAdio-Detection-And-Ranging} (RADAR) and \emph{LIght-Detection-And-Ranging} (LIDAR). They face the challenges of long computation latency caused by complex signal processing and computer vision, potential inaccuracy of identifying TVs from their background environments, and ineffectiveness in the presence of blockages between a TV-SV pair. Their drawbacks have contributed to many accidents involving self-driving cars. To overcome these limitations of current technologies, we present in this paper a new technology {for relative positioning at the SV}, called \emph{Cooperative Multi-point Positioning} (COMPOP). Essentially, by detecting the cooperative signals broadcast by multi-antennas distributed at a TV, the SV estimates the antenna positions representing the TV skeleton, thereby performing multi-point TV positioning. Algorithms based on \emph{Fourier transform} (FT) are proposed for fast and accurate COMPOP not only when there are \emph{lines-of-sight} (LoS) but also when they are blocked. The latter exploits the signals reflected by the surfaces of nearby vehicles.	

\subsection{Single-Point Positioning}
Single-point positioning techniques were originally developed for locating mobile devices and recently also applied to autonomous driving. The most popular and simplest is to use a built-in \emph{Global Positioning System} (GPS) receiver for computing the receiver's position and sharing the information to peer devices over wireless links. However, in an urban environment, the required LoS links between GPS receivers and satellites are often blocked by e.g., buildings or tunnels. This issue has
motivated researchers to develop alternative techniques relying on base stations or access points in wireless networks as anchors to estimate the position of a mobile by either triangularization or measuring signal power~\cite{SynGap}. This requires the mobile to estimate anchors' positions from their signals. The negative effect of mobility on estimation accuracy can be coped with by utilizing sampled temporal measurements and motion models~\cite{MobileCentric}. Nevertheless, due to the unreliability of fading channels, such positioning techniques assisted by a wireless network cannot reach the level of precision required for autonomous driving. The required high-resolution positioning can be realized by \emph{ultra-wideband} (UWB) radios leveraging the fact that a large bandwidth overcomes multi-path fading and thereby enables accurate \emph{time-of-arrival} (ToA) measurements~\cite{UWBPR}.
To avoid the need of transmitter-receiver clock synchronization, \emph{time-difference-of-arrival} (TDoA)
and \emph{phase-difference-of-arrival} (PDoA) based methods are proposed which cancel at a receiver the clock difference of a transmitter (anchor) by observing received signals from different antennas or frequencies~\cite{TDoA,hyperbolic}. Most recently, TDoA-based positioning over \emph{non-LoS} (NLoS) links is made possible by separating multi-paths and locating the source via exploiting the paths' geometric relation~\cite{Kaifeng}. The deployment of such techniques enables the positioning of a hidden TV (one without an LoS link). Despite having a rich literature, the outputs of single-point positioning are insufficient for complex maneuvers in autonomous driving such as platooning and overtaking, for which multi-point TV positioning is required.

\subsection{Vehicular Sensing} \label{intro:sensing}
Vehicular sensing can be treated as an extreme form of multi-point positioning. The sensing
process involves scanning the surrounding environment and then recognizing, imaging, and positioning objects useful for autonomous driving (such as nearby pedestrians and vehicles). {Usually vehicular sensing does not rely on the GPS link, which detects relative locations of objects independently with an on-board sensing system.} Relevant technologies can be grouped as passive or active, depending on whether they require radiation. 
Typical passive sensors include infrared sensors, cameras, and passive \emph{millimeter waves} (mmWave) sensors. They exploit ambient and unintended infrared radiation, light, and mmWave to image the sources and reflectors~\cite{PassiveMMW}. The two most popular types of active vehicular sensors are RADAR and LIDAR~\cite{Husheng,BackPropagation}. A RADAR scans the environment by steering a microwave beam using an antenna array and
observes reflected signals with varying attenuation to image the environment ~\cite{SFCW2016,Advanced,MIMO1}. Subsequently, targeted objects are detected and positioned using signal processing and computer vision. LIDAR operates based on a similar principle except for replacing the microwave
beam with a mechanically steered sharp laser beams, thereby achieving a higher resolution~\cite{LIDARscanning}. Among the sensing technologies, those based on light (i.e., cameras) or infrared (e.g., infrared sensors) are exposed to severe performance degradation caused by hostile weather such as heavy rains and thick fogs~\cite{LIDARweather}. Overall, existing multi-point positioning technologies share two drawbacks that present key challenges for vehicular sensing and positioning, as described below. 
\begin{itemize}
	\item {\bf Latency and accuracy:} Environmental scanning (for RADAR and LIDAR) and object recognition (for RADAR, LIDAR, and cameras) are time-consuming. For example, a 3D beam scanning by a RADAR using a large-scale phased array can incur around ten-second latency~\cite{MIMO1}. This is unacceptable for autonomous driving in a crowded urban environment or at high speeds. On the other hand, objective recognition using a well-trained deep neural network with an onboard GPU typically takes several seconds. Furthermore, existing objective recognition techniques relying on offline training are easily affected by a variation on object features. As a result, the detection accuracy is usually in the range of $70\% \sim 90\%$~\cite{Training1}. Their application to auto-driving presents safety threats, as exemplified by recent fatal accidents.
	
	\item {\bf Hidden vehicle detection:} Besides TVs in sight, detecting hidden vehicles in the sensing blind
	spots (e.g., a TV around a street corner) can avoid many potential accidents. Though some progress has been made on NLoS single-point TV positioning, the desired detection of hidden TVs as multi-point objects is still an uncharted area and the theme of this work.
\end{itemize}  	

%{\blue Moreover, one straightforward approach for multi-point positioning is to position the antennas broadcasting signals on the TV sequentially by the aforementioned single-point positioning techniques at the receiver. However, such a design does not make the best of the resources, leading to performance degradation compared with vehicular sensing techniques, i.e., RADAR, LIDAR and the proposed COMPOP.  Sequential single-point positioning in time causes severe latency and detection error due to the vehicle movement. Parallel single-point positioning in time by using orthogonal frequencies results in low resolution because the detection of each point cannot exploit the whole system bandwidth, otherwise interference will disable the detection.}

\subsection{Main Contributions}
Multi-point TV positioning refers to positioning a TV as a multi-point object where each point corresponds to a transmit antenna. As an attempt to tackle the two challenges discussed in the preceding section, we propose in this work the framework of \emph{cooperative multi-point positioning} (COMPOP) at a SV building on the cooperation that a TV broadcasts a signal with a waveform that facilitates relative TV positioning at the SV. In this framework, latency reduction is achieved in two ways.
\begin{itemize}
	\item The first is to retrieve from the total received signal the desired TV signal for further processing using the embedded signature. This avoids the time-consuming conventional method of 	environmental scanning and learning-based TV identification.
	\item The second is to apply the low-latency FT to efficiently compute a power spectrum distributed over the 3D space. This allows direct positioning by peak detection.
\end{itemize}
To tackle the second challenge of multi-point positioning without an LoS link, we propose a novel technique of using reflected TV signals over the smooth surfaces of nearby vehicles to position multiple ``virtual TVs''. {Then we exploit their geometry for combining to position the actual TV without any priori knowledge of the nearby vehicles. By overcoming these two key challenges to the existing positioning approaches, we believe that the proposed technique can be well integrated with the current positioning system for improving the safety of autonomous driving.}

The specific designs of the proposed COMPOP framework are summarized as follows.
\begin{itemize}
	\item {\bf COMPOP over a LoS link:} The framework for this case comprises two key algorithms operating
	in separate bandwidths. {First, due to transceiver separation, clock synchronization between the transceivers are necessary for coherent signal demodulation at the SV. By using two reserved single tones and given the knowledge of TV
	signature waveform, we propose to use a PDoA-based iterative algorithm for TV-SV synchronization
	by accurately estimating their clock difference from the received TV signals in the presence of
	channel noise.} Second, the remaining bandwidth is used by the TV to transmit a conventional
	\emph{stepped-frequency-continuous-wave} (SFCW), a multi-tone waveform with a uniform frequency gap between tones \cite{SFCW2016}. After correcting the clock difference, a determined relation between the TV position and the received SFCW signal is established. As a main feature of COMPOP, we propose the application of FT to transform the received SFCW
	signal into the mentioned spectrum over the 3D space for estimating the TV position by peak
	detection. Besides, the positioning accuracy is analyzed.
	\item {\bf COMPOP over mirror-reflection links:} For the case without LoS, but with nearby vehicles as reflectors, a COMPOP technique is developed as follows. The combined use of antenna array and TV signature allows the SV to resolve the
	received signals as reflected by different smooth vehicular surfaces. The application of the
	preceding technique for LoS COMPOP on the resolved signals yields multi-point positions of
	multiple ``virtual vehicles''. {Without priori knowledge of the nearby vehicles' locations, an intelligent combining approach is proposed to position the actual TV by exploiting the geometric relation between the virtual and actual TVs.}
\end{itemize}
 
The remainder of this paper is organized as follows. In Section II, the system model and signal model are introduced. In Section III and IV, the proposed COMPOP is elaborated and analyzed in both LoS and NLoS conditions, and a systematic comparison between the proposed and existing techniques is presented. The realistic simulation results are presented in Section V, and the conclusion follows in Section VI.

\section{System Model}
	
\begin{figure}[t]\vspace{-0.5cm}
	\begin{minipage}[t]{0.5\linewidth}%设定图片下字的宽度，在此基础尽量满足图片的长宽
		\centering
		\includegraphics[scale = 0.55]{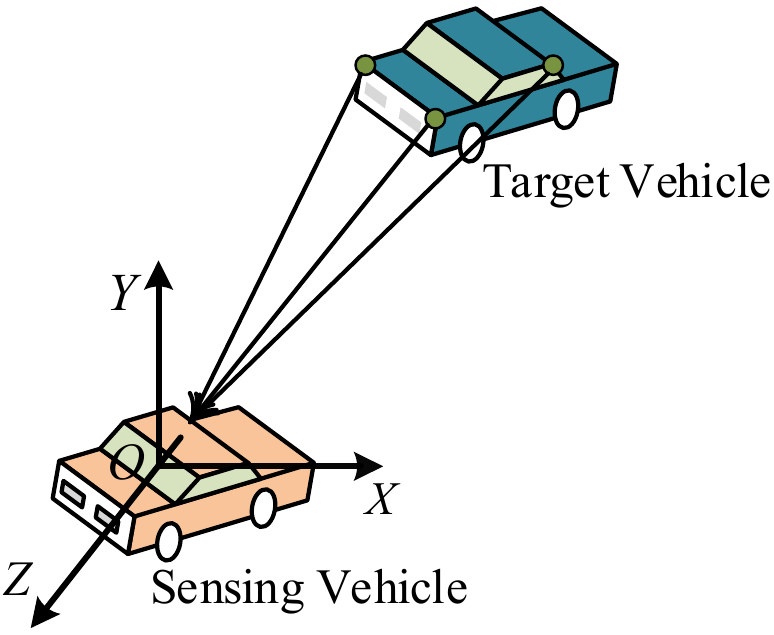}
		\caption*{(a) COMPOP over a LoS link.}%加*可以去掉默认前缀，作为图片单独的说明
		\label{Direct}
	\end{minipage}
	\begin{minipage}[t]{0.5\linewidth}%需要几张添加即可，注意设定合适的linewidth
		\centering
		\includegraphics[scale = 0.42]{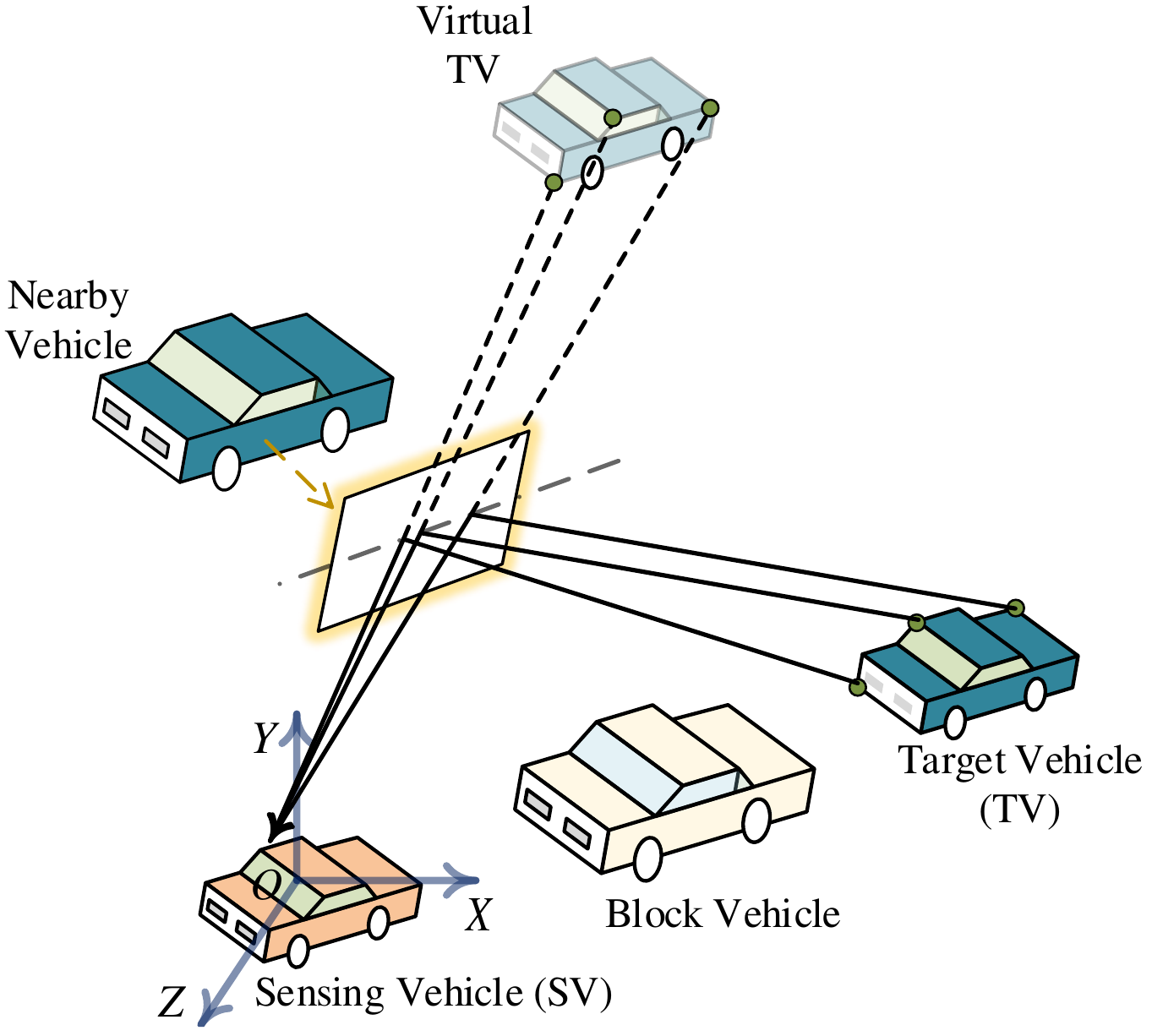}
		\caption*{(b) COMPOP over a mirror-reflection link.}
		\label{Indirect}
	\end{minipage}
	\caption{Two scenarios of COMPOP.}\label{mirror}\vspace{-0.5cm}
\end{figure}
	
Consider the scenario with multiple vehicles on the road, including one pair of TV and SV. The TV  is equipped with an antenna array with $N_t$ elements distributed over the vehicle body such that the distribution sketches its shape. Consequently, the SV equipped with a distributed  array of $N_r$ antennas   performs multi-point positioning of the TV by locating its antennas from their broadcast waveforms in a mmWave spectrum. As illustrated in Fig.~\ref{mirror}, we consider both the scenarios where the TV and SV are connected with an LoS link or NLoS links. For the scenario of NLoS links  [see Fig. \ref{mirror}(b)], COMPOP at the SV relies on signals reflected on the smooth surfaces of nearby vehicles. Though mmWave signals are attenuated severely by scattering and reflection on non-smooth surfaces, the attenuation is found to be small if the surfaces are smooth, e.g., those of vehicles' polished metallic body~\cite{28GHz}. 
%For the convenience  of defining the coordinate system, we assume without loss of generality  that the receive  aperture is located at the SV's left side (see Fig. \ref{mirror}). Then 
%Without loss of generality, the $X$, $Y$, and $Z$-axes of the coordinate system are defined as the SV's moving direction, elevation, and depth, respectively.

\subsection{Signal  Models}\label{SignalModel}
Each TV antenna transmits the superposition of two waveforms for facilitating different operations of COMPOP. One is  a  signature waveform (a multi-tone waveform) transmission (similar to that in \cite{SignatureWaveform}) enabling the SV to estimate the SV-TV system clock difference so as to eliminate its negative effect on the positioning. The second is a SFCW waveform (also a multi-tone waveform) in a separate bandwidth from the signature, which facilitates COMPOP at the SV. The waveforms are described as follows.

\subsubsection{Signature Waveform} For clock-difference estimation, it is sufficient to transmit different signature waveforms over  two antennas, whose  indices are denoted as $\mathsf{a}$ and $\mathsf{b}$, using two single tones for each antenna. The antenna coordinates are represented as  $\mathbf{x}_{\mathsf{a}} = (x_{\mathsf{a}}, y_{\mathsf{a}}, z_{\mathsf{a}})$ and $\mathbf{x}_{\mathsf{b}} = (x_{\mathsf{b}}, y_{\mathsf{b}}, z_{\mathsf{b}})$, respectively. The two  waveforms for the two antennas, denoted as  $s_{\mathsf{a}}$ and $s_{\mathsf{b}}$, are given as 
\begin{align}\label{Signature_Waveform}
s_{\mathsf{a}}(t+\sigma)=e^{j2\pi f_{\mathsf{a}} (t+\sigma)}+ e^{j2\pi (f_{\mathsf{a}}+\Delta) (t+\sigma)}, \qquad
s_{\mathsf{b}}(t+\sigma)= e^{j2\pi f_{\mathsf{b}} (t+\sigma)}+ e^{j2\pi (f_{\mathsf{b}}+\Delta) (t+\sigma)},
\end{align}
where $f_{\mathsf{a}}$ and $f_{\mathsf{b}}$ are two orthogonal frequencies specifying the signatures, $\Delta$ is a given  frequency separation, and  $\sigma$ is the TV-SV system clock difference (in sec). 

First, consider the scenario with $L$ surface-reflection links $(L \ge 1)$. 
Let $\Gamma^{(\ell)}$ denote the complex reflection coefficient of the $\ell$-th link given as $\Gamma^{(\ell)}=|\Gamma^{(\ell)}|\exp(j \angle \Gamma^{(\ell)})$. Moreover,  let $\tau_{n,m}^{(\ell)}$ denote  the signal flight  time from TV's  antenna $n$ to SV's antenna $m$ proportional to the propagation distance $d_{n,m}^{(\ell)}$, i.e., $d_{n,m}^{(\ell)}=c\cdot \tau_{n,m}^{(\ell)}$ where $c$ is the speed of light. { Given the notations above, the raw received waveform at each SV antenna, say antenna $m$, is given as
\begin{align}
W_m = \sum\limits_{\ell = 1}^L w_m^{(\ell)}(t),
\end{align}
where $w_m^{(\ell)}(t) = \Gamma^{(\ell)} s_{\mathsf{a}}( t + \sigma-{\tau _{\mathsf{a},m}^{(\ell)}}) + \Gamma^{(\ell)} s_{\mathsf{b}}( t + \sigma-{\tau _{\mathsf{b},m}^{(\ell)}})$.
We assume that the AoAs of different signal arrivals are separable using a classic technique, e.g., MUSIC~\cite{MUSIC}, where the AoAs can be accurately detected in the angular domain by searching the power spectrum of the received signals~\footnote{The effect of antennas on AoA detection accuracy is well studied by simulations in~\cite{AoA} such that AoA error is significantly reduced as the number of antenna increases.}.  %Therefore, we consider perfect AoA detection in the following presentation. 
Specifically, after AoA detection, the received signals can be differentiated by coherent detection in the angular domain, and regrouped as a vector ${\mathbf{w}}_m = [{w}_m^{(1)}(t), \cdots, {w}_m^{(L)}(t)]^T$. Note that each element in ${\mathbf{w}}_m$ comprises waveforms at four different frequencies.}
By exploiting frequency orthogonality, different frequency components in the received waveform antenna $m$, can be separated and grouped to form two $K$ by $L$ matrices to facilitate the algorithmic design in the sequel:
\begin{align}\label{signature waveform:MatrixForm}
{\boldsymbol A}_{m}(t)=[{\boldsymbol a}_{m}^{(1)}(t), \cdots,{\boldsymbol a}_{m}^{(L)}(t)],
\quad 
{\boldsymbol B}_{m}(t)=[{\boldsymbol b}_{m}^{(1)}(t), \cdots, {\boldsymbol b}_{m}^{(L)}(t)], 
\end{align}
where
\begin{align} 
{\boldsymbol{a}}_{m}^{(\ell)}(t)=\Gamma^{(\ell)}\left[\begin{matrix}e^ { j{2\pi {f_{\mathsf{a}}}\left( t+\sigma-{\tau _{\mathsf{a},m}^{(\ell)}} \right)}} \\ 
e^{ j{2\pi (f_{\mathsf{a}}+\Delta)\left( {t+\sigma-{\tau _{\mathsf{a},m}^{(\ell)}}} \right)} } \end{matrix}\right], \qquad
{\boldsymbol{b}}_{m}^{(\ell)}(t) =\Gamma^{(\ell)}\left[\begin{matrix}e^ { j{2\pi {f_{\mathsf{b}}}\left( t+\sigma-{\tau _{\mathsf{b},m}^{(\ell)}} \right)}} \\ 
e^{ j{2\pi (f_{\mathsf{b}}+\Delta)\left( {t+\sigma-{\tau _{\mathsf{b},m}^{(\ell)}}} \right)} } \end{matrix}\right]. 
\end{align}
Next, two matched filters are designed as ${\bf D}_{\mathsf{a}}= \mathsf{diag}\left\{e^{-j2\pi f_{\mathsf{a}}t}, e^{-j2\pi (f_{\mathsf{a}}+\Delta)t}\right\}$ and 
${\bf D}_{\mathsf{b}}= \mathsf{diag}\left\{e^{-j2\pi f_{\mathsf{b}} t},\right.\\ \left. e^{-j2\pi (f_{\mathsf{b}}+\Delta) t} \right\}$, for demodulating  ${\boldsymbol A}_{m}(t)$ and ${\boldsymbol B}_{m}(t)$ from time functions into matrix  symbols:
\begin{align}
 \label{signature waveform_A}  {\mathbf{A}}_{m} &= {\bf D}_{\mathsf{a}}{\mathbf{A}}_{m}(t)=\left[ {\boldsymbol{a}}_{m}^{(1)}, {\boldsymbol{a}}_{m}^{(2)}, \cdots,  {\boldsymbol{a}}_{m}^{(L)}\right]\\
 \label{signature waveform_B}  {\mathbf{B}}_{m}&= {\bf D}_{\mathsf{b}}{\mathbf{B}}_{m}(t)=\left[ {\boldsymbol{b}}_{m}^{(1)}, {\boldsymbol{b}}_{m}^{(2)}, \cdots,  {\boldsymbol{b}}_{m}^{(L)}\right], 
\end{align}
where
\begin{align} \label{Coefficient_signature waveforms}
{\boldsymbol{a}}_{m}^{(\ell)} =  \Gamma^{(\ell)}\left[\begin{matrix}e^{ j{2\pi {f_{\mathsf{a}}}\left( \sigma-{\tau _{\mathsf{a},m}^{(\ell)}} \right)}}\\ 
e^{ j{2\pi (f_{\mathsf{a}}+\Delta)\left( {\sigma-{\tau _{\mathsf{a},m}^{(\ell)}}} \right)} }\end{matrix}\right],\qquad 
{\boldsymbol{b}}_{m}^{(\ell)} =  \Gamma^{(\ell)}\left[\begin{matrix}e^{ j{2\pi {f_{\mathsf{b}}}\left( \sigma-{\tau _{\mathsf{b},m}^{(\ell)}} \right)}}\\ 
e^{ j{2\pi (f_{\mathsf{b}}+\Delta)\left( {\sigma-{\tau _{\mathsf{b},m}^{(\ell)}}} \right)} }\end{matrix}\right].
\end{align}
{ It is worthwhile to notice that all equations above also hold if one of the $L$ links is a LoS link.}

Next, consider the other scenario with a LoS link to the TV, { where unresolvable reflection links possibly exist but are neglected due to the significant power difference between LoS and NLoS paths.} %Note that for reflection on a planar and smooth vehicular surface, the reflection coefficient can be approximated as a constant for all antennas  \cite{SFCW2016,Antenna}). 
Without loss of generality, let the resultant uniform  channel gains be normalized as $\Gamma=1$.  { Then the LoS counterparts of  ${\mathbf{A}}_{m}$ and  ${\mathbf{B}}_{m}$ can be simplified from \eqref{signature waveform_A} and \eqref{signature waveform_B} as: 

\begin{align} \label{Coefficient_signature waveforms_LoS}
{\mathbf{A}}_{m}^{\text{LOS}} = \left[\begin{matrix}e^{ j{2\pi {f_{\mathsf{a}}}\left( \sigma-{\tau _{\mathsf{a},m}} \right)}}\\ 
e^{ j{2\pi (f_{\mathsf{a}}+\Delta)\left( {\sigma-{\tau _{\mathsf{a},m}}} \right)} }\end{matrix}\right],\qquad 
{\mathbf{B}}_{m}^{\text{LOS}} = \left[\begin{matrix}e^{ j{2\pi {f_{\mathsf{b}}}\left( \sigma-{\tau _{\mathsf{b},m}} \right)}}\\ 
e^{ j{2\pi (f_{\mathsf{b}}+\Delta)\left( {\sigma-{\tau _{\mathsf{b},m}}} \right)} }\end{matrix}\right]. 
\end{align} 
where $\tau_{n,m}$ denotes the signal flight time from TV's antenna $n$ to SV's antenna $m$ proportional to the propagation distance $d_{n,m}$ in LoS.}
	
\subsubsection{SFCW Waveform}

A multi-tone waveform commonly used in RADAR, called  \emph{stepped-frequency-continuous-wave} (SFCW)~\cite{SFCW2016, BackPropagation,NearFieldPhasedArray}, is broadcast by each TV antenna. { The waveform, denoted as ${{s}}(t)$, comprises multiple single-tone continuous-waves  with equally separated frequencies by a fixed frequency gap $\Delta$. Mathematically, 
\begin{align}\label{transmitScalar}
{{s}}(t)= \sum\limits_{k=1}^{K}\exp(j{2\pi {f_k}{t}}),
\end{align}
where $\mathcal{F} = \{f_k\}_{k=1}^K$ represents the set of frequencies  such that $f_k=f_1+(k-1)\Delta$ for $k=1, \cdots, K$. The vector form of the SFCW is also provided as
\begin{align}\label{transmitSig}
{\boldsymbol{s}}(t)= \left[\exp(j{2\pi {f_1}{t}}), ..., \exp(j{2\pi {f_K}{t}}) \right]^T,
\end{align}
since the single-tone continuous-waves are naturally separated in frequency.}

First, consider the scenario with $L$ surface-reflection links $(L \ge 1)$. { With frequency decoupling,} the received signal at the SV's antenna $m$ is 
\begin{align}\label{Eq:TotReceivedSig}
{\boldsymbol r}_{m}(t)=\sum_{\ell=1}^L  {\boldsymbol r}_{m}^{(\ell)}(t),
\end{align}
where ${\boldsymbol r}_{m}^{(\ell)}(t) \in \mathbb{C}^{K\times1}$ represents the signal vector from the $\ell$-th surface-reflection link as 
\begin{align}\label{receiveSig}
{\boldsymbol r}_{m}^{(\ell)}(t)&=\Gamma^{(\ell)}\sum_{n=1}^{N_t} {\boldsymbol r}_{n,m}^{(\ell)}(t)=\Gamma^{(\ell)}\sum_{n=1}^{N_t} {\boldsymbol s}(t+\sigma-\tau_{n,m}^{(\ell)}). 
\end{align}
{ By assuming perfect AoA detection using the MUSIC technique~\cite{MUSIC}, we can decompose \eqref{Eq:TotReceivedSig} into individual ${\boldsymbol r}_{m}^{(\ell)}(t)$ by coherent detection in the angular domain, which is rewritten as a $K$ by $L$ matrix ${\mathbf{R}}_{m}(t)$ as follows:
\begin{align}\label{Eq:MatrixForm}
{\mathbf{R}}_{m}(t)=\left[ {\boldsymbol{r}}_{m}^{(1)}(t), \cdots,  {\boldsymbol{r}}_{m}^{(L)}(t)\right].
\end{align} }
Assuming that the TV-SV system clock difference $\sigma$ is estimated  as $\tilde{\sigma}$, the received signal \eqref{Eq:MatrixForm} is demodulated by multiplying the $K$ by $K$ matched filtering matrix ${\bf D}= \mathsf{diag}\{{\boldsymbol{s}}(t+\tilde{\sigma})^{H}\}$ as
\begin{align}\label{Demod}
{\mathbf{Y}}_m=\left[ {\mathbf{y}}_{m}^{(1)}, \cdots,  {\mathbf{y}}_{m}^{(L)}\right]
={\bf D}{\mathbf{R}}_{m}(t),
\end{align}
where
${\mathbf{y}}_{m}^{(\ell)} = {\bf D} {\mathbf{r}}_{m}^{(\ell)}(t)= \left[y_m^{\ell,1}, y_m^{\ell,2},..., y_m^{\ell,K}\right]^T$ with the component $y_m^{\ell,k}$ being
\begin{align}
y_m^{\ell,k} &= \Gamma^{(\ell)}\sum_{n=1}^{N_t}{{\exp\left[ { j2\pi {f_k}(\sigma-\tilde{\sigma}-{\tau_{n,m}^{(\ell)}} )  } \right]}}. 
\end{align}
With accurate estimation of the clock difference using the algorithm in Section \ref{sec:Synchronization}, the $(\ell,k)$-th received signal component is simplified as
\begin{align}\label{demodulated}
y_m^{\ell,k} &= \Gamma^{(\ell)}\sum_{n=1}^{N_t}{{\exp\left[ { -j2\pi {f_k}{\tau_{n,m}^{(\ell)}}  } \right]}}, \ \forall \ell, k.
\end{align}

Next, consider the other scenario with a LoS link. The received signal component in \eqref{demodulated} can be further simplified with $\Gamma = 1$ as 
\begin{align}\label{demodulated_LoS}
y_m^{k} &= \sum_{n=1}^{N_t}{{\exp\left[ { -j2\pi {f_k}{\tau_{n,m}}  } \right]}}, \ \forall k. 
\end{align}

Moreover, without loss of generality, we define the origin $O$ of the coordinate system at the center of the SV, $Z$-axis in the direction of AoA. The $X$, $Y$-axes are parallel and vertical to the ground, respectively, perpendicular to the $Z$-axis, as shown in Fig. \ref{mirror}.

%It is shown that the demodulation signal $y_m^{\ell,k}$ \eqref{demodulated} totally depends on $\{\tau^{(\ell)}_{n,m}\}$ 
%if the synchronization gap is perfectly estimated ($\sigma=\tilde{\sigma}$). Otherwise, it is distorted. 
\begin{remark}[Feasible Ranging Distance of SFCW]\label{feasibleDis}Due to the periodicity of phases, the maximum ranging distance of  SFCW should be limited by $R_{\max}=\frac{c}{\Delta}$ (in meters) to avoid ambiguity \cite{SFCW2016}.% The frequency gap $\Delta$ is commonly no smaller than $10$ MHz due to the  hardware limitation, and the resultant feasible distance is up to about $30$ m. 
\end{remark}

\begin{remark}[Channel Fading]\label{SaomplingProcess}  The sampling process takes $T_s = \frac{1}{\mathcal{B}}$ seconds, where $\mathcal{B}$ is the system bandwidth at the SV. Since the sampling duration ($3.3^{-10}$ s with $\mathcal{B}=3$ GHz) is much shorter than the coherence time (approximately $T_c = 1$ ms at $30$ GHz  with velocity $v = 10$ m/s~\cite{Goldsmith}, the channel fading, comprising the pathloss and small-scale fading, is considered as a constant and omitted in the presentation for convenience.% The frequency gap $\Delta$ is commonly no smaller than $10$ MHz due to the  hardware limitation, and the resultant feasible distance is up to about $30$ m. 
\end{remark}

%{
%\begin{remark}[Doppler Effect Compensation]\label{Doppler}
%As a result of  the relative velocity between the TV and SV, a Doppler shift may cause severe main lobe spreading  in AoA detection and hence the loss of \emph{signal-to-noise-ratio} (SNR). Some existing technologies have been used for velocity compensation and Doppler shift mitigation. For example, by detecting the main lobe width of the SFCW's range profile, a technique for velocity estimation and compensation approach with a high accuracy is proposed in~\cite{Doppler}. In view of existing effective techniques, we assume perfect Doppler effect compensation in this paper.
%\end{remark} }

%	{\zzz
%    \begin{remark}
%    	Synchronization can be achieved only when the SFCW frequency gap $\Delta$ is divisable by the frequency separation $\Delta_f$ of a signature waveform pair.
%    \end{remark}
    %\begin{proof}
    %	Please refer to Appendix A.
	%\end{proof}

	\subsection{Procedure Design}\label{Sec:PF}
	
	For their non-overlapping spectrums, the two waveforms can be separated by the SV using filtering. Moreover, the signal processing delay at the SV is fixed and known from calibration, allowing the suppression of  its effect on TV positioning.

	\subsubsection{Clock Synchronization} To establish direct relation between the TV's multi-point position and the superimposed signal $y_m^{\ell,k}$ \eqref{demodulated} or $y_m^{k}$ \eqref{demodulated_LoS}, the SV aims at compensating the system clock difference $\sigma$ by using  \eqref{signature waveform_A}, \eqref{signature waveform_B} or \eqref{Coefficient_signature waveforms_LoS}, divided into two cases as follows.
	\begin{itemize}
		{ \item {\bf A LoS link:} When a LoS link exists, the SV is expected to detect the system clock difference $\sigma$ from the received signature waveforms $\{{\mathbf{A}}_{m}^{\text{LOS}}, {\mathbf{B}}_{m}^{\text{LOS}} |\ \forall m\}$,
		where ${\mathbf{A}}_{m}^{\text{LOS}}, {\mathbf{B}}_{m}^{\text{LOS}}$ are matrix signals in \eqref{Coefficient_signature waveforms_LoS} received over the $m$-th antenna.
		
		\item {\bf Reflection links:} The clock synchronization needs to be achieved for each reflection link. Take the $\ell$-th reflection link as an example, the clock difference $\sigma$ needs to be detected based on signature waveforms $\{{\boldsymbol{a}}_{m}^{(\ell)}, {\boldsymbol{b}}_{m}^{(\ell)} |\ \forall m\}$, where ${\boldsymbol{a}}_{m}^{(\ell)}, {\boldsymbol{b}}_{m}^{(\ell)}$ are given in \eqref{Coefficient_signature waveforms}.
		}
	\end{itemize}
	
	 %Assume that the synchronization is made. The  $y_m^{\ell,k}$ \eqref{demodulated} is then rewritten by the following surface integral form:
	%\begin{align}\label{demodulated_sync}
	%y_m^{\ell,k} \underset{(\sigma=\tilde{\sigma})}{\Longrightarrow} \Gamma^{(\ell)}\int {\mathsf{I}_\mathbf{x^{(\ell)}}{\exp\left( - j2\pi\frac{f_k}{c} D(\mathbf{x}, \mathbf{p}_{m}) \right)}d{\mathbf{x}}},
	%\end{align}
	%where $\mathsf{I}_\mathbf{x^{(\ell)}}$ is an indicator to become $1$ if a transmit antenna exists on point $\mathbf{x}$, which is symmetric to the point ${\mathbf{x^{(\ell)}}}$ w.r.t. the reflection surface $\ell$, and $0$
	%otherwise, and 
	\subsubsection{Multi-point Positioning}
	Define the set of transmit-antenna locations at a TV as ${\mathcal X} = \{\mathbf{x}_n\}$ with $\mathbf{x}_n \in \mathbb{R}^3$ and $\left|\mathcal{X}\right| = {N_t}$. Then, an indicator function $\mathsf{I}_{\{\mathbf{x} \in \mathcal{X}\}}$ to represent the transmit antenna's distribution over the 3D spatial domain is defined as
	\begin{align}
	\mathsf{I}_{\{\mathbf{x} \in \mathcal{X}\}} = \sum\limits_{n=1}^{N_t}\delta(\mathbf{x}-\mathbf{x}_n),
	\end{align}
    where $\delta(\cdot)$ is a delta function satisfying $\delta(\mathbf{x}) = 0, \forall \mathbf{x} \ne \mathbf{0}$ and $\int_{\mathbb{R}^3}\delta(\mathbf{x})d\mathbf{x}=1$. Let $D(\mathbf{x}_n, \mathbf{p}_{m}) = \|\mathbf{x}_n - \mathbf{p}_{m}\|$ measures the LoS distance between $\mathbf{x}_n$ and the location of receive antenna $m$ denoted by $\mathbf{p}_m$. For reflection links,  $D^{(\ell)}{(\mathbf{x}_n, \mathbf{p}_{m})}$ measures the propagation distance over the $\ell$-th link. Accordingly, we directly have $\tau_{n,m} = \frac{D(\mathbf{x}_n, \mathbf{p}_{m})}{c}$ and $\tau_{n,m}^{(\ell)} = \frac{D^{(\ell)}{(\mathbf{x}_n, \mathbf{p}_{m})}}{c}$. Also, the received SFCW waveforms \eqref{demodulated} and \eqref{demodulated_LoS} can be rewritten as 
	\begin{align}\label{reflectionLink}
	y_m^{\ell,k} &= \Gamma^{(\ell)}\sum_{n=1}^{N_t}{{\exp\left( { - j\frac{2\pi {f_k}}{c}D^{(\ell)}{(\mathbf{x}_n, \mathbf{p}_{m})}   } \right)}}, \ \forall \ell, k.
	\end{align}
	for reflection links and
	\begin{align}\label{LoSLink}
	y_m^{k} &= \sum_{n=1}^{N_t}{{\exp\left( {- j\frac{2\pi {f_k}}{c}D{(\mathbf{x}_n, \mathbf{p}_{m})}  } \right)}}, \ \forall k. 
	\end{align}
	for the LoS link if it exists.
	
	As aforementioned, the point set $\mathcal{X}$ represents the multi-point position information of the TV. { Estimating $\mathcal{X}$ can be divided into two cases as follows.
	\begin{itemize}
		\item {\bf A LoS link:} The multi-point TV position $\mathcal{X}$ is directly retrieved from SFCW signals received in LoS, i.e., $\{{y}_m^k |\ \forall m,k \}$.
	    \item {\bf Reflection Links:} The real multi-point TV position $\mathcal{X}$ needs to be detected by combining all SFCW signals from different reflection links, i.e., $\{\mathbf{y}_m^{(\ell)} |\ \forall m, \ell \}$.	
	\end{itemize} }

%	In summary, the proposed COMPOP follows three steps: 1) synchronization between the TV and SV by solving \ref{SyncProblemFormulation}, 2) multi-point detection of position $\ell$  by solving \ref{ProblemDefinition1}, and 3) combining multiple virtual TVs to position the actual TV by solving \ref{ProblemDefinition2}. 
		
	\begin{figure}[t]\vspace{-0.5cm}
		\centering
		\includegraphics[%height=240pt, 
		width=350pt]{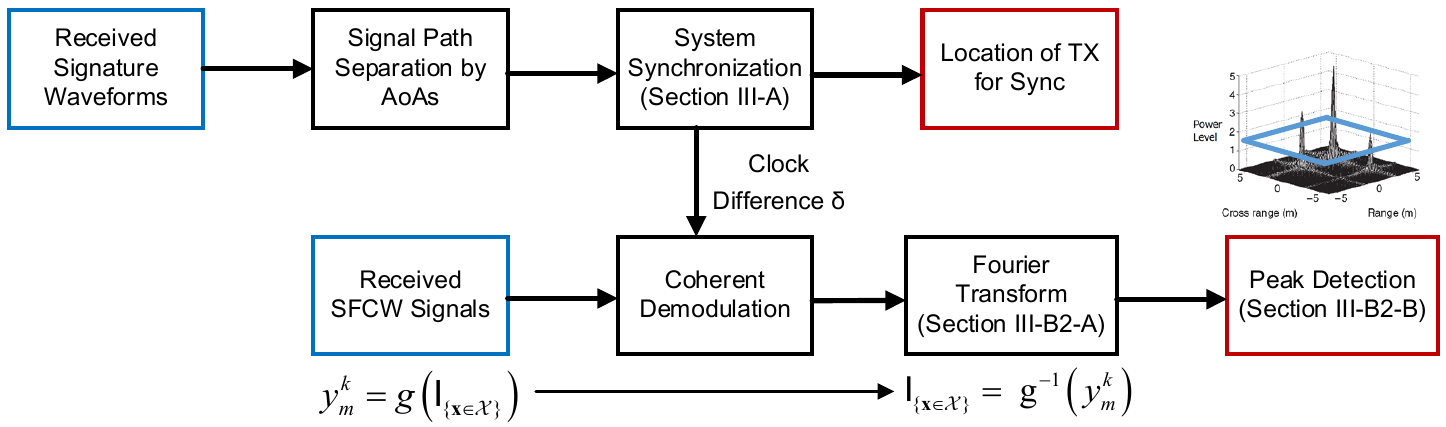}
		\caption{Diagram of COMPOP in LoS.}\label{LoS_Diag}\vspace{-0.7cm}
	\end{figure}

	\section{COMPOP over a LoS Link}\label{sec:LoS}
	In this section, we consider a case where a LoS link between the TV and the SV exists, making it reasonable to ignore other reflection links due to the significant power difference between LoS and NLoS paths.
	The scheme we propose for LoS case is illustrated in Fig. \ref{LoS_Diag}, consisting of two steps: 1) synchronization; and 2) multi-point positioning. The overview and algorithm description for each step are presented in the following. 
	
	{ For clarification, we firstly summarized the assumptions used in this section:
	\begin{itemize}
		\item {\bf Perfect AoA Detection:} Accurate AoA detection can be realized by applying classical MUSIC algorithm~\cite{MUSIC} with an appropriate number of antennas at the receiver.
		\item {\bf Constant Channel Fading:} The channel fading is considered constant during the sampling process.
		%\item {\bf Perfect Doppler Effect Compensation:} Doppler effect compensation can be realized by existing techniques for the signature waveforms and SFCW signals as elaborated in Remark \ref{Doppler}.
	\end{itemize} }

	\subsection{Step 1: Synchronization}\label{sec:Synchronization}
	%\begin{figure}
	%	\centering
	%	\includegraphics[%height=220pt, 
	%	width=250pt]{Figures/syn.pdf}
	%	\caption{Illustration of the synchronization approach.}\label{Syn}
	%\end{figure}
\subsubsection{Overview} 
As illustrated in Sec. \ref{SignalModel}, synchronization is necessary to compensate the system clock difference $\sigma$, enabling correct coherent demodulation \eqref{Demod} at the receiver and the subsequent multi-point positioning. Note that the clock difference is contained in the signal phases in \eqref{Coefficient_signature waveforms_LoS}, which can be observed at the receiver. However, to estimate the clock difference $\sigma$ directly is challenging because it is coupled with the propagation delay $\tau_m$ as shown in \eqref{Coefficient_signature waveforms_LoS}. Therefore, it is necessary to separate the two parameters by estimating the propagation delay first, which can be translated into the estimation of one representative transmit antenna, say $\mathbf{x}_\mathsf{a}$, given the knowledge of the receive antennas' locations. To this end, we first estimate $\mathbf{x}_\mathsf{a}$ in the presence of noise by applying a \emph{phase-difference-of-arrival} (PDoA) based method~\cite{Newton,hyperbolic}.

\subsubsection{Algorithm Description}
Here we give an approach to estimate the location $\mathbf{x}_\mathsf{a}$ of the representative transmit antenna $\mathsf{a}$ from the received signature waveform $\mathbf{A}_m^{\text{LOS}}$ in \eqref{Coefficient_signature waveforms_LoS}. The index of the transmit antenna are omitted for brevity, i.e. $\tau_m = \tau_{{\mathsf{a}}, m}$.
\begin{itemize}
%	\item \emph{ Principle:} Recalling the relation between the propagation distance $D(\mathbf{x}_{\mathsf{a}}, \mathbf{p}_{m})$ and $\tau_m$, i.e., $D(\mathbf{x}_{\mathsf{a}}, \mathbf{p}_{m})=\tau_m\cdot c$ with light speed $c$, the system clock difference $\sigma$ is given as
%	\begin{align}\label{Sigma}
%	\sigma=\tau_m-\frac{\eta_m}{2\pi\Delta}=\frac{D(\mathbf{x}_{\mathsf{a}}, \mathbf{p}_{m})}{c}-\frac{\eta_m}{2\pi\Delta} = \frac{\|\mathbf{x}_{\mathsf{a}} -  \mathbf{p}_{m}\|}{c}-\frac{\eta_m}{2\pi\Delta},
%	\end{align}
%	where $\mathbf{x}_{\mathsf{a}}$ is the location of the antenna $\mathsf{a}$ on the TV, and $\eta_{ m}=2\pi {\Delta}({\tau _{m}} -\sigma)$ is the phase difference between the two components of the received signature waveform at receive antenna $m$. According to \eqref{Sigma}, given the phase difference $\eta_m$ measured at the receiver, the clock difference can be estimated if $\mathbf{x}_{\mathsf{a}}$ is known. Therefore, we give a design to estimate the clock difference $\sigma$ by detecting the TX location $\mathbf{x}_{\mathsf{a}}$.	
	\item \emph{Detection of TX Location:}   
	Let $\mathsf{F}_{m}(\mathbf{x}_{\mathsf{a}})$ denote the propagation distance difference from the antenna ${\mathsf{a}}$ to the SV's antennas $m$ and $1$, given as 
	\begin{align}\label{FuncDefine}
	\mathsf{F}_{m}(\mathbf{x}_{\mathsf{a}}) = D(\mathbf{x}_{\mathsf{a}}, \mathbf{p}_{m}) - D(\mathbf{x}_{\mathsf{a}}, \mathbf{p}_{1}), \quad m=2,\cdots, {N_r}.
	\end{align}
	%Note that $\mathsf{F}_{m}(\mathbf{x}_{\mathsf{a}})$ is a univariate function w.r.t. the TX location $\mathbf{x}_{\mathsf{a}}$.
    At the SV's antenna $m$, the phase difference $\eta_{ m}=2\pi {\Delta}({\tau _{m}} -\sigma)$ between the two components of the received signature waveforms can be directly measured from ${\mathbf{A}}_{m}^{\text{LOS}}$ in \eqref{Coefficient_signature waveforms_LoS} with noise in the presence, denoted as $\tilde{\eta}_m$. Then based on the relation $\eta_{ m}=2\pi {\Delta}({\tau _{m}} -\sigma) = 2\pi\Delta\left(\frac{D(\mathbf{x}_\mathsf{a}, \mathbf{p}_m)}{c}-\sigma \right)$, a noisy measurement of  $\mathsf{F}_{m}(\mathbf{x}_{\mathsf{a}})$ is given as
	\begin{align}\label{Measurements}
	\tilde{\mathsf{F}}_{m} = c\frac{\left( \tilde{\eta} _{m} - \tilde{\eta} _{1} \right)}{2\pi \Delta } =
	 c\frac{\left( \eta _{m} - \eta _{1} \right)}{2\pi \Delta } + \Delta \phi_m, \quad m=2,\cdots, {N_r},
	\end{align}
    where $\Delta \phi_m$ is the additional Gaussian noise. Since $\mathsf{F}_{m}(\mathbf{x}_{\mathsf{a}})$ is univariate w.r.t. the location $\mathbf{x}_{\mathsf{a}}$, the TX location can be optimized based on the \emph{minimum-mean-square-error} (MMSE) criterion to minimize the gap between $\mathsf{F}_{m}(\mathbf{x}_{\mathsf{a}})$ and the measurement $\tilde{\mathsf{F}}_{m}$ as follows.
    \begin{Problem}[TX Location Optimzation]\label{LocOpt}
    	\begin{align}\label{Optimization}
    	\mathop {\min }\limits_{\mathbf{x}_{\mathsf{a}}} \quad \sum\limits_{m = 1}^{N_r} {\left\| {{{\tilde {\mathsf{F}}}_m} -{\mathsf{F}_m}({{\bf{x}}_{\mathsf{a}}}) } \right\|_2^2}.
    	\end{align}
    \end{Problem}
    {    Then the optimal solution ${\mathbf{x}_{\mathsf{a}}}^*$, which naturally coincides with the actual location $\mathbf{x}_{\mathsf{a}}$, can be achieved by using the iterative Gauss-Netwon method \cite{Gauss_Newton} with an arbitrary initial point $\tilde{\mathbf{x}}_{\mathsf{a}}$ as follows
    \begin{align}\label{Solution}
    	\begin{gathered}
    		\boldsymbol{h} = {\left( {{{\bf{G}}(\tilde{\mathbf{x}}_{\mathsf a})^T}{\bf{G}}(\tilde{\mathbf{x}}_{\mathsf{a}})} \right)^{ - 1}}{{\bf{G}}(\tilde{\mathbf{x}}_{\mathsf{a}})^T}{\bf{b}}(\tilde{\mathbf{x}}_{\mathsf{a}})\\
    		{\tilde{\mathbf{x}}_{\mathsf{a}}} \longleftarrow {\tilde{\mathbf{x}}_{\mathsf{a}}} + \boldsymbol{h},
    	\end{gathered}
    \end{align}
    where
    \begin{align}
    	{\bf{G}}(\tilde{\mathbf{x}}_{\mathsf{a}})=\left[ {\begin{array}{*{20}{c}}
    			{\frac{{\partial {\mathsf{F}_{2}}({\tilde{\mathbf{x}}_{\mathsf{a}}})}}{{\partial x_{\mathsf{a}}}}}&{\frac{{\partial {\mathsf{F}_{2}}({\tilde{\mathbf{x}}_{\mathsf{a}}})}}{{\partial y_{\mathsf{a}}}}}&{\frac{{\partial {\mathsf{F}_{2}}({\tilde{\mathbf{x}}_{\mathsf{a}}})}}{{\partial z_{\mathsf{a}}}}}\\
    			{...}&{...}&{...}\\
    			{\frac{{\partial {\mathsf{F}_{{N_r}}}({\tilde{\mathbf{x}}_{\mathsf{a}}})}}{{\partial x_{\mathsf{a}}}}}&{\frac{{\partial {\mathsf{F}_{{N_r}}}({\tilde{\mathbf{x}}_{\mathsf{a}}})}}{{\partial y_{\mathsf{a}}}}}&{\frac{{\partial {\mathsf{F}_{{N_r}}}({\tilde{\mathbf{x}}_{\mathsf{a}}})}}{{\partial z_{\mathsf{a}}}}}
    	\end{array}} \right], \qquad
    	{\bf{b}}(\tilde{\mathbf{x}}_{\mathsf{a}})=
    	\left[ {\begin{array}{*{20}{c}}
    			{{{\tilde {\mathsf{F}}}_2} - {\mathsf{F}_{2}}({\tilde{\mathbf{x}}_{\mathsf{a}}})}\\
    			{...}\\
    			{{{\tilde {\mathsf{F}}}_{N_r}} - {\mathsf{F}_{{N_r}}}({\tilde{\mathbf{x}}_{\mathsf{a}}})}
    	\end{array}} \right].
    \end{align}
}

    \item \emph{Detection of Propagation Time:} With the knowledge of ${\mathbf{x}}_{\mathsf{a}}$, the propagation time $\tau_m$ can be detected for signals at receive antenna $m$ by $\tau_m = \frac{D(\mathbf{x}_\mathsf{a}, \mathbf{p}_{m})}{c}$.
    
    \item \emph{Subtraction \& Averaging:} By subtracting the propagation time $\tau_m$ component from the noisy phase measurement $\tilde{\eta}_m$, the system clock difference can be differently calculated depending on the choice of the SV's antenna $m$, denoted by $\tilde \sigma_m = \tau_m-\frac{\tilde{\eta}_m}{2\pi\Delta}$. Averaging these values gives an accurate estimate of $\sigma$ such that $\tilde \sigma=\sum_{m=1}^{N_r} \tilde \sigma_m$.
\end{itemize}
\begin{prop} [Synchronization Feasibility Condition]\label{Prop1}
	\emph{At least four SV's antennas are required (${N_r}\geq 4$) to detect the TV-SV system clock difference, according to the solution of \eqref{Optimization} in~\cite{Gauss_Newton}.}
\end{prop}

\begin{figure}[t]\vspace{-0.5cm}
	\centering
	\includegraphics[%height=240pt, 
	width=280pt]{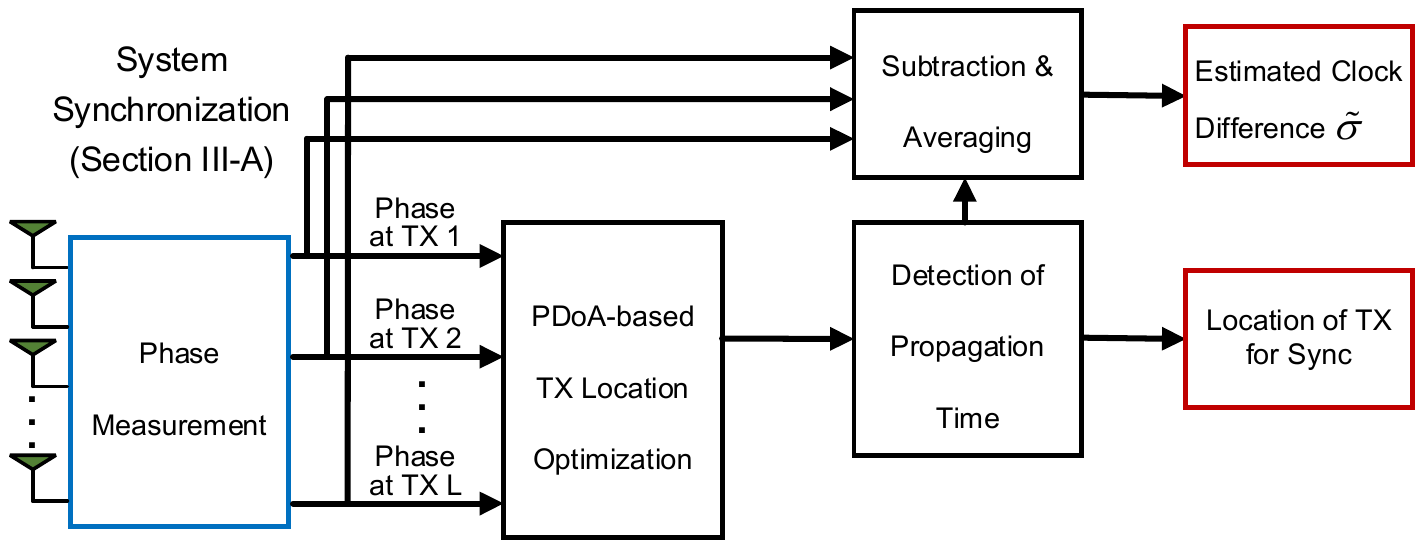}
	\caption{Diagram of the system clock synchronization process.}\label{Sync_Diag}\vspace{-0.8cm}
\end{figure}

    Note that although such an approach is based on a similar principle to the method in~\cite{hyperbolic}, we give a new design to output a clock difference estimation. Moreover, one more antenna location ${\mathbf{x}}_{\mathsf{b}}$ can be detected by applying the same algorithm above on the signature waveform ${\mathbf{B}}_{m}^{\text{LOS}}$ in \eqref{Coefficient_signature waveforms_LoS}  at the receiver. The detected locations ${\mathbf{x}}_{\mathsf{a}}$, ${\mathbf{x}}_{\mathsf{b}}$ of the representative transmit antennas help the COMPOP over the reflection links illustrated later in Section \ref{sec2:Common-point}.

\begin{remark}[TV Recognition in LoS]\label{TVrecognition}
   In LoS case, the SV is able to resolve signals from different TVs according to the AoAs with an antenna array~\cite{MUSIC}. Considering the signature waveform and SFCW transmissions share the same signal paths, the clock difference detected from signature waveforms will be used for SFCWs with the same AoAs, and the multi-point TV position detected in Sec. \ref{sec:Imaging} will be mapped to the same AoAs as well.
\end{remark}
{
\begin{remark}[Initial Value Selection]{It is recommended to use a solution satisfying any three equations in
		\begin{align}\label{Initial}
		\tilde{\mathsf{F}}_{m} = D(\mathbf{x}_{\mathsf{a}}, \mathbf{p}_{m}) - D(\mathbf{x}_{\mathsf{a}}, \mathbf{p}_{1}), \quad m=2,\cdots, {N_r}.
		\end{align}
		 as the initial selection of ${\tilde{\mathbf{x}}_{\mathsf{a}}}$, where the convergence to the global optimal is verified by simulations. }
\end{remark}
}
\begin{remark}[Sampling Requirement]\label{remark1}
		The synchronization procedures are based on the assumption that the phase gap estimated at each two adjacent receive antennas is no larger than $\frac{\pi}{2}$. Thus the distance between each two adjacent receive antenna at the SV needs to satisfy $\Delta_d < \frac{c}{{2\Delta }}$. %This requirement can be easily satisfied in practice.
\end{remark}

{
\begin{prop}[Error Covariance]\label{Prop:sync} \emph{We use the error covariance matrix ${\rm{cov}}({\bf{\hat x}}_{\mathsf a})$ as the performance metric of the above synchronization approach, defined as
	\begin{align}\label{ErrorCov}
	{\rm{cov}}({\bf{\tilde x}}_{\mathsf a}) ={{\mathbb E}}\left[\left({\bf{\tilde x}}_{\mathsf a}
	-\mathbb{E}[{\bf{\tilde x}}_{\mathsf a}]\right)
	\left({\bf{\tilde x}}_{\mathsf a}-\mathbb{E}[{\bf{\tilde x}}_{\mathsf a}]\right)^T\right].
	\end{align}
	For tractability, we 
	assume that the phase error follows an \emph{independent identically distributed} (i.i.d.) Gaussian distribution where $\Delta \phi_m \sim  \mathcal{N}(0, \sigma_z^2{\bf I}), \forall m$. As the number of SV's antennas $N_r$ becomes larger, the covariance matrix ${\rm{cov}}({\bf{\tilde x}}_{\mathsf a})$ scales with $\mathcal{O}\left(\frac{1}{N_r-1}\right)$.  }
\end{prop}
\begin{proof}
	See Appendix A.
\end{proof}

According to Proposition \ref{Prop:sync}, we assume perfect synchronization in the following steps for convenience. However, in simulations, we keep the phase noise and the resultant clock difference detection error through the entire process.

\begin{remark}[Comparison with Existing Synchronization Methods]
	%By exploiting the signature waveform transmission from the TV to SV, the proposed synchronization method accurately detects the clock difference at the SV.	
	Such a clock synchronization method clearly differs from conventional GPS-based approaches~\cite{GPSsync} since it does not rely on the GPS. Compared to consensus-based synchronization approaches~\cite{YC}, the method we proposed can be processed in real time and thus is more suitable for vehicular sensing.
\end{remark}
}

%\begin{remark}[Initial Value Selection]{The Gauss-Newton method sometimes converges to a local optimal point due to the wrong selection of the initial  ${\tilde{\mathbf{x}}_{\mathsf{a}}}$~\cite{Newton}. It is recommended to use a solution satisfying any three equations among \eqref{FuncDefine} as its initial selection, of which the convergence of the global optimal is verified by simulation. }
%\end{remark}

	\subsection{Step 2: Multi-Point Positioning}\label{sec:Imaging}
	\subsubsection{Overview}
  	As shown in Fig. \ref{LoS_Diag}, the system clock difference $\sigma$ can be removed by the preceding step, facilitating the following procedures. The main idea of the multi-point positioning step in LoS is briefly illustrated as follows. We first show that the received signal $y_m^{k}$ can be presented as a function of the indicator $\mathsf{I}_{\{\mathbf{x} \in \mathcal{X}\}}$ which represents the power spectrum of the transmit antennas, denoted as $y_m^{k} = g(\mathsf{I}_{\{\mathbf{x} \in \mathcal{X}\}})$. Therefore, the position information $\mathcal{X}$ of the TV can be retrieved from the received signals through an inverse function and a following peak detection. We give the estimation of $\mathcal{X}$ in the following algorithm, divided into two phases 1) Fourier transform; and 2) peak detection, as described in Fig. \ref{LoS_Diag}.
  	
  	\subsubsection{Algorithm Description}
  	  	\begin{table}[t] \vspace{-0.5cm} 
  		\small
  		\caption{Fourier Transform Pairs}  
  		\begin{center}  
  			\begin{tabular}{|l|l|}  
  				\hline  
  				Fourier Transform & Inverse Fourier Transform  \\ \hline  
  				$ \mathsf{FT}_{\mathsf{2D}}\left( \mathsf{h}(\mathbf{x})\right)  = \int_{\mathbb{R}^2}{\mathsf{h}(\mathbf{x})\exp{(-j\frac{2\pi}{c}\mathbf{f}^T\mathbf{x})}dxdy}$ & $ \mathsf{FT}_{\mathsf{2D}}^{-1}\left( \mathsf{H}(\mathbf{f})\right)  = \int_{\mathbb{R}^2}{\mathsf{H}(\mathbf{f})\exp{(j\frac{2\pi}{c}\mathbf{f}^T\mathbf{x})}df^{(x)}df^{(y)}}$  \\ \hline
  				$\mathsf{FT}_{\mathsf{3D}}\left( \mathsf{h}(\mathbf{x})\right)  = \int_{\mathbb{R}^3}{\mathsf{h}(\mathbf{x})\exp{(-j\frac{2\pi}{c}\mathbf{f}^T\mathbf{x})}d\mathbf{x}}$ & $\mathsf{FT}_{\mathsf{3D}}^{-1}\left( \mathsf{H}(\mathbf{f})\right)  = \int_{\mathbb{R}^3}{\mathsf{H}(\mathbf{f})\exp{(j\frac{2\pi}{c}\mathbf{f}^T\mathbf{x})}d\mathbf{f}}$  \\ \hline  
  				\tabincell{l}{$\mathsf{FT}_{\mathsf{2D}}^D(\{\mathsf{h}[x,y]\})$\\ $= \sum\limits_{n_x = -\infty}^{\infty}\sum\limits_{n_y = -\infty}^{\infty}\!\!\!\mathsf{h}[x,y]\exp\!\left( \!-j\frac{2\pi}{c}( {f^{(x)}n_x}\!+\!{f^{(y)}n_y})\right)$} & \tabincell{l}{${\mathsf{FT}_{\mathsf{2D}}^D}^{-1}(\mathsf{H}(\mathbf{f}))$\\ $= \int_{\mathbb{R}^2}\!\mathsf{H}(\mathbf{f})\exp\!\left( j\frac{2\pi}{c}\!\left( {f^{(x)}n_x}\!+\!{f^{(y)}n_y}\right)\!\right)\!df^{(x)}df^{(y)}$}  \\ \hline 
  				\tabincell{l}{$\mathsf{FT}_{\mathsf{3D}}^D(\{\mathsf{h}[x,y,z]\}) $\\$= \sum\limits_{n_x = -\infty}^{\infty}\sum\limits_{n_y = -\infty}^{\infty}\sum\limits_{n_z = -\infty}^{\infty}\mathsf{h}[x,y,z]\exp\left( -j\frac{2\pi}{c}\mathbf{f}^T\mathbf{n_x}\right)$} & \tabincell{l}{${\mathsf{FT}_{\mathsf{3D}}^D}^{-1}(\mathsf{H}(\mathbf{f}))= \int_{\mathbb{R}^3}\mathsf{H}(\mathbf{f})\exp\left( j\frac{2\pi}{c}\mathbf{f}^T\mathbf{n_x}\right)d\mathbf{f}$} \\ \hline 
  				\multicolumn{2}{|l|}{\tabincell{l}{$\mathbf{x}=(x,y,z)^T$ represents a vector in 3D space, and $\mathbf{f}=\left( f^{(x)},f^{(y)},f^{(z)}\right)^T $ represents a vector in 3D frequency\\ domain, and $\mathbf{n_x} = (n_x,n_y,n_z)^T$ denotes the index of the sample $\mathsf{h}[x,y,z]$ in the set $\{\mathsf{h}[x,y,z]\}$.}}\\ \hline
  			\end{tabular}  
  		\end{center}  \vspace{-0.5cm} \label{TableFourier}
  	\end{table} 
  	 The synchronized demodulation \eqref{demodulated_LoS} enables to express  ${y}_m^{k}$ as \eqref{LoSLink}, which can be rewritten in a 3D surface integral form as
  	\begin{align}\label{imageSignal}
  	{y}_m^{k} = {\mathsf{y}}(\mathbf{p}_m, f_k)%\nonumber \\
  	= \int_{\mathbb{R}^3} {\mathsf{I}_{\{\mathbf{x} \in \mathcal{X}\}}{\exp\left( - j\frac{2\pi f_k}{c}D{(\mathbf{x}, \mathbf{p}_{m})} \right)}d{\mathbf{x}}}, %\nonumber \\ 
  	\end{align}
  	{ where ${\mathsf{y}}(\mathbf{x}, f): \mathbb{R}^4\to\mathbb{R}$ is a continuous function,} and $D{(\mathbf{x}, \mathbf{p}_{m})} = \|\mathbf{x}-\mathbf{p}_m\|$ represents the Euclidean distance between point $\mathbf{x}=(x,y,z)^T$ and the location of the SV's antenna $m$, denoted by $\mathbf{p}_m = (p_m^{(x)},p_m^{(y)},p_m^{(z)})^T$.  Then based on \eqref{imageSignal}, we have the following lemma based on the scalar diffraction idea in \cite{AngularSpectrum}.
  	\begin{lemma}\label{lemma1}
  		Consider the indicator function of  transmit antennas $\mathsf{I}_{\{\mathbf{x} \in \mathcal{X}\}}$ and the function $\mathsf{y}\left(\mathbf{x}, f\right)$ representing signals at the receiver. The following equality holds in the frequency domain as
  		\begin{align}\label{freq_eq}
  		\mathsf{FT}_{\mathsf{2D}}\left(\mathsf{s}(x,y,f)\right) \big|_{f=\|\mathbf{f}\|} = \mathsf{FT}_{\mathsf{3D}} (\mathsf{I}_{\{\mathbf{x} \in \mathcal{X}\}}).
  		\end{align}
  		where $\mathbf{f} = \left(f^{(x)}, f^{(y)}, f^{(z)}\right)^T$ is a spatial frequency vector, and $\mathsf{s}(x,y,f) = {\mathsf{y}}(x,y,0,f)$ is a continuous function, representing received signals at the $X-Y$ plane $z=0$. The $\mathsf{FT}_{\mathsf{2D}}(\cdot)$ and $\mathsf{FT}_{\mathsf{3D}}(\cdot)$ are 2D and 3D Fourier transforms defined in Table \ref{TableFourier}.
  	\end{lemma}
    \begin{proof}
    	Please refer to Appendix B.
    \end{proof}

    \begin{remark}[Sampling at the Receiver]
    	Based on the received signals, we are only able to collect samples of $\mathsf{y}(\mathbf{x},f)$ discretely over the receive antenna array. Recall that $Z$-axis is in the direction of AoA known at the SV, leading to the following approximation
    	\begin{align}\label{approxSample}
    	\mathsf{s}(p_m^{(x)},p_m^{(y)},f_k) = {{\mathsf{y}}}\!\left(p_m^{(x)},p_m^{(y)},0, f_k\right) &\approx{{\mathsf{y}}}\left(\mathbf{p}_m, f_k\right)\exp\left(\!-j\frac{2\pi f_k}{c}p_m^{(z)}\!\right) ,
    	\end{align}
    	which holds tightly when the TV-SV distance is much larger than the SV's size.	
    	Therefore, based on the approximation in \eqref{approxSample}, discrete samples on $\mathsf{s}(x,y,f)$ are also available.
    \end{remark}

    \begin{remark}[Linear Interpolation]
    	Since the receive antennas may not be regularly distributed, a linear interpolation \cite{interpolation} is adopted to provide samples of $\mathsf{s}(x,y,f)$ uniformly spaced along the $X$ and $Y$ directions. Take samples at two adjacent receive antennas in the $X$ direction as an example, whose coordinates are supposed to be $\mathbf{p}_{m_1}$ and $\mathbf{p}_{m_2}$, and ${p}_{m_1}^{(y)}={p}_{m_2}^{(y)}$. { Then the interpolation generates samples between these two points as
    	\begin{align}\label{interpolation}
    	\mathsf{s}(x,p_{m_1}^{(y)},f) = \frac{x-p_{m_1}^{(x)}}{p_{m_2}^{(x)}-p_{m_1}^{(x)}}\mathsf{s}(p_{m_2}^{(x)},p_{m_2}^{(y)},f)+\frac{p_{m_2}^{(x)}-x}{p_{m_2}^{(x)}-p_{m_1}^{(x)}}\mathsf{s}(p_{m_1}^{(x)},p_{m_1}^{(y)},f), \qquad p_{m_1}^{(x)} <x< p_{m_2}^{(x)}.
    	\end{align} }
    	By using interpolation in \eqref{interpolation} along  $X$ and $Y$ directions sequentially, continuous samples of $\mathsf{s}(x,y,f)$ are generated on the plane $z=0$, and uniformly spaced samples are available as well.
    	The samples of $\mathsf{s}(x,y,f)$ after the interpolation are denoted as $\{\mathsf{s}[x,y,f]\}$. Since SFCW signals are used, the samples are naturally discrete and uniformly spaced in frequency.  	    	
    \end{remark}
    { It is also worthwhile to notice that common interpolation methods, e.g., linear, spline and polynomial interpolations are all capable of accurate resampling here. Therefore, we simply adopt linear interpolation for low-complexity and the interpolation error is omitted in the following presentation for convenience. }

  	Moreover, to preserve complete information in the frequency domain with discrete samples, the distances between two adjacent receive antennas along $X$ or $Y$ directions are both required to be less than $\frac{c}{2(f_1+f_K)}$ according to the Nyquist sampling criterion, which will be elaborated later in Sec. \ref{sec:res}.  	
  	 Here we simply assume that the samplings at the receiver satisfy the Nyquist criterion and the interpolation is accurate. Then the LHS of \eqref{freq_eq} with discrete inputs turns to be 
  	\begin{align}\label{outLeft}
  	\mathsf{FT}_{\mathsf{2D}}^{D}\left(\{ \mathsf{s}[x,y,f]\}\right) = \mathsf{S}(f^{(x)}, f^{(y)}, f),
  	\end{align} 
  	where $\mathsf{FT}_{\mathsf{2D}}^{D}$ refers to the 2D discrete-time Fourier transform defined in Table \ref{TableFourier}, and the output $\mathsf{S}(f^{(x)}, f^{(y)}, f)$ is a function continuous in $f^{(x)}$ and $f^{(y)}$ domain while discrete in $f$ direction. By considering $\mathsf{S}(f^{(x)}, f^{(y)}, f)$ as a function of $\mathbf{f}$ according to the relation ${f=\|\mathbf{f}\|}$, the continuous function $\left\{\mathsf{S}(\mathbf{f})\big| \forall \mathbf{f} \in \mathbb{R}^3\right\}$ over the 3D frequency domain can be estimated by an interpolation \cite{interpolation}. Then the LHS of \eqref{freq_eq} can be approximated as
  	\begin{align}\label{approximation}
  	\mathsf{FT}_{\mathsf{2D}}\left(\mathsf{s}(x,y,f)\right)\big|_{f=\|\mathbf{f}\|} \approx \mathsf{L}\bigg(\mathsf{S}(f^{(x)}, f^{(y)}, f)\big|_{f=\|\mathbf{f}\|}\bigg) = \mathsf{L}\bigg(\mathsf{FT}_{\mathsf{2D}}^{D}\left( \{\mathsf{s}[x,y,f]\}\right)\big|_{f=\|\mathbf{f}\|}\bigg),
  	\end{align}  	
  	where  $\mathsf{L}(\cdot)$ represents the linear interpolation process \cite{interpolation}. 
  	
  	On the other hand, the RHS of \eqref{freq_eq} can be represented in the discrete case as $\mathsf{FT}_{\mathsf{3D}}^D\left( \{[\mathsf{I}_{\{\mathbf{x} \in \mathcal{X}\}}]\}\right)$, where $\mathsf{FT}_{\mathsf{3D}}^D$ is the 3D discrete-time Fourier transform. Then according to \eqref{approximation}, the equality in \eqref{freq_eq} directly gives
  	\begin{align}\label{discreteapprox}
  	\mathsf{L}\bigg(\mathsf{FT}_{\mathsf{2D}}^D\left( \{\mathsf{s}[x,y,f]\}\right)\big|_{f=\|\mathbf{f}\|}\bigg) \approx \mathsf{FT}_{\mathsf{3D}}^D\left(  \{[\mathsf{I}_{\{\mathbf{x} \in \mathcal{X}\}}]\}\right).
  	\end{align}
  	The approximation in \eqref{discreteapprox} comes from the interpolation process, which is proved to be accurate by simulations.
    \begin{itemize}
	\item \emph{A. Fourier Transform:}  According to \eqref{discreteapprox}, the indicator function $ \mathsf{I}_{\{\mathbf{x} \in \mathcal{X}\}}$ can be estimated in a discrete form via a 3D inverse discrete-time Fourier transform as
	\begin{align}
	\Phi[{\mathbf{x}}] =  {\mathsf{FT}_{\mathsf{3D}}^{D}}^{-1}\left\{\mathsf{L}\bigg(\mathsf{FT}_{\mathsf{2D}}^D\left(\{ \mathsf{s}[x,y,f]\}\right)\big|_{f=\|\mathbf{f}\|}\bigg) \right\} \approx  \{[\mathsf{I}_{\{\mathbf{x} \in \mathcal{X}\}}]\} .
	\end{align}

  		\item \emph{B. Peak Detection:} { Locations of the transmit antennas can be estimated by peak detection over the power spectrum $\{|\Phi[{\mathbf{x}}]|\}$ after a power normalization, namely,
  		\begin{align}\label{Mapping_Rule}
  		 \mathcal{X} = \bigg\{\mathbf{x}\ \bigg|\ \frac{|\Phi[\mathbf{x}]|}{\{|\Phi[\mathbf{x}]|\}_{\text{max}}} \geq\nu\bigg\},
  		\end{align}
  		where $\nu$ represents the detection threshold appropriately selected, and $\{|\Phi[\mathbf{x}]|\}_{\text{max}}$ is the maximum value of the detected power spectrum $\{|\Phi[{\mathbf{x}}]|\}$. }
  	\end{itemize}
    Therefore, according to the algorithm above, the multi-point TV position can be retrieved from  signals over the antenna array at the receiver, for the LoS case.

	%\begin{remark}[Resampling]\label{resample} {To calculate the inverse 3D discrete Fourier transform in \eqref{resonstruction}, sampling on frequency domain with constant  interval is necessary. 
	%		However, due to the nonlinear relation of each frequency component $f \in \mathcal F$ as ${f^{(z)} = \sqrt{{{f}^2} - \left(f^{(x)}\right)^2 - \left(f^{(y)}\right)^2}}$, regular samplings on $f^{(x)}$ and $f^{(y)}$ lead to the irregular sequence on $f^{(z)}$ domain. It is thus required to make the sequence regular by using an interpolation, which is called a \emph{resampling}. We use a linear interpolation for the resampling whose error is marginal verified by simulation. }
	%\end{remark}

%\begin{remark}[Fast Data Collection]{In the multi-point positioning scheme, all TVs around transmit positioning waveforms possibly at the same time. With signal differentiation according to AoAs, the SV collects the waveforms from all TVs at the same time, and retrieve the multi-point position information using the FFT-based algorithm proposed above. Thereby, coupled with instantaneous TV recognition in Remark \ref{TVrecognition}, ultra-low latency can be achieved by the proposed cooperative vehicular positioning.}
%\end{remark}
	
\subsection{Resolution Analysis}\label{sec:res} 
This subsection provides analysis on resolution of the multi-point position retrieved by the above algorithm. Here we consider the 3D antenna array at the receiver as an `equivalent aperture' located in the $X-Y$ plane $z=0$, as shown in Fig. \ref{resolution}(a).

For ease of understanding, we first introduce the following terminologies.

\begin{definition}[Bandwidth]\label{Def:Bandwidth} {A bandwidth $\mathcal{B}$ in a direction of the frequency domain (at the receiver) is defined as the maximum frequency difference of the received signal along the corresponding direction in the spatial domain.} %In this paper, the bandwidth in a certain frequency direction refers to the maximum frequency (phase changing rate) range of the received signals w.r.t. the distances in the corresponding direction.}
\end{definition}

%We will show later that the bandwidth in $f^{(z)}$ direction is directly $\mathcal{B}_z = f_K-f_1$ corresponding to the SFCW signals. On the other hand, the bandwidths in $f^{(x)}$ and $f^{(y)}$ directions are $\mathcal{B}_x=\mathcal{B}_y \approx f_c$, where $f_c = \frac{f_1+f_K}{2}$ is the center frequency of the SFCW signals.

\begin{definition}[Resolution]{A resolution $\delta$, representing the positioning accuracy, is defined as the minimum distance to differentiate multiple objects, TV's antennas in our work. The resolution is said to be better when the minimum distance is smaller. %{ Recall that the $X$, $Y$, and $Z$-axes represent the SV's moving direction, elevation, and depth.} 
Accordingly, resolutions in $X$ or $Y$ directions and $Z$ direction are respectively called the \emph{azimuth} and \emph{range} resolutions. }
\end{definition}

The direct relation between    
bandwidth $\mathcal B$ and resolution $\delta$ is established as 
\begin{align}
\delta=\frac{c}{\mathcal{B}},
\end{align}
where $c$ is the light speed. Based on this relation, the azimuth and range resolutions are firstly analyzed, and then sampling requirements to achieve a given resolution  are provided next. 
\subsubsection{Azimuth Resolution} 
Suppose the detected TV position $\mathcal{X}$ of \eqref{Mapping_Rule}  is projected on the equivalent aperture's center denoted by $m_c$ as shown in Fig.~\ref{resolution}(b). We focus on the spatial resolution in $Y$ direction since resolutions in $X$ and $Y$ directions become equivalent when the aperture is a $D$-by-$D$ square. 
Consider signals at frequency $f_k$, the phase difference between received signal at $m_1$ and $(m_1-\Delta_y)$ is approximately $\frac{D}{4\sqrt{R^2+(D/2)^2}}\frac{2\pi f_k}{c}\Delta_y$ w.r.t. a subtle distance $\Delta_y$, corresponding to frequency $\frac{D}{\sqrt{4R^2+D^2}}{f_k}$. Moreover, the phase difference between received signal at $(m_c+\Delta_y)$ and $m_c$ is $0$, corresponding to frequency $0$.  According to Definition \ref{Def:Bandwidth}, the bandwidth in $f^{(y)}$ direction, denoted by ${\cal B}_y$ is approximately
\begin{align}
{\cal B}_y \approx 2\times \mathbb{E}_k \left[\frac{D}{{\sqrt {{4{ {{R}} }^2} + {D^2}} }} f_k-0 \right]=\frac{2D}{{\sqrt {{4{ {{R}} }^2} + {D^2}} }} f_c,
\end{align}
where $f_c=\frac{f_1+f_K}{2}$. 
The azimuth resolution in $Y$ direction can be straightforward obtained as
\begin{align}\label{spatialResolution}
\delta _y = \frac{c}{\mathcal{B}_y} \approx \frac{{c\sqrt {{4{ {{R}} }^2} + {D^2}} }}{{2{f_c}D}}.
\end{align}
It is shown to be proportional to the term $\frac{{\sqrt {{4{{{R}} }^2} + {D^2}} }}{D}$, meaning that higher azimuth resolution can be achieved when the TV is closer and the aperture size becomes larger.		

\subsubsection{Range Resolution} 
Samplings at a fixed location $\left(p_m^{(x)}, p_m^{(x)}, 0\right)$ consist of signals at all frequencies used.
Thus bandwidth in the $f_z$ direction is 
$\mathcal{B}_z \approx (f_K-f_1)$, and the range resolution can be obtained as
\begin{align}\label{RangeResolution}
{\delta _z} = \frac{c}{\mathcal{B}_z}\approx \frac{{c }}{{\left( {{f_K} - {f_1}} \right)}},
\end{align}
where $f_K$ and $f_1$ are maximum and minimum frequencies in the SFCW specified in \eqref{transmitSig}.

%\begin{figure}[t]\vspace{-0.5cm}
%	\centering
%	\includegraphics[%height=240pt, 
%	width=200pt]{Figures/resolution1.pdf}
%	\caption{Illustration of azimuth and range resolutions.}\label{resolution}\vspace{-0.5cm}
%\end{figure}

\begin{figure}[t]\vspace{-0.5cm}
	\begin{minipage}[t]{0.5\linewidth}%设定图片下字的宽度，在此基础尽量满足图片的长宽
		\centering
		\includegraphics[height = 90pt]{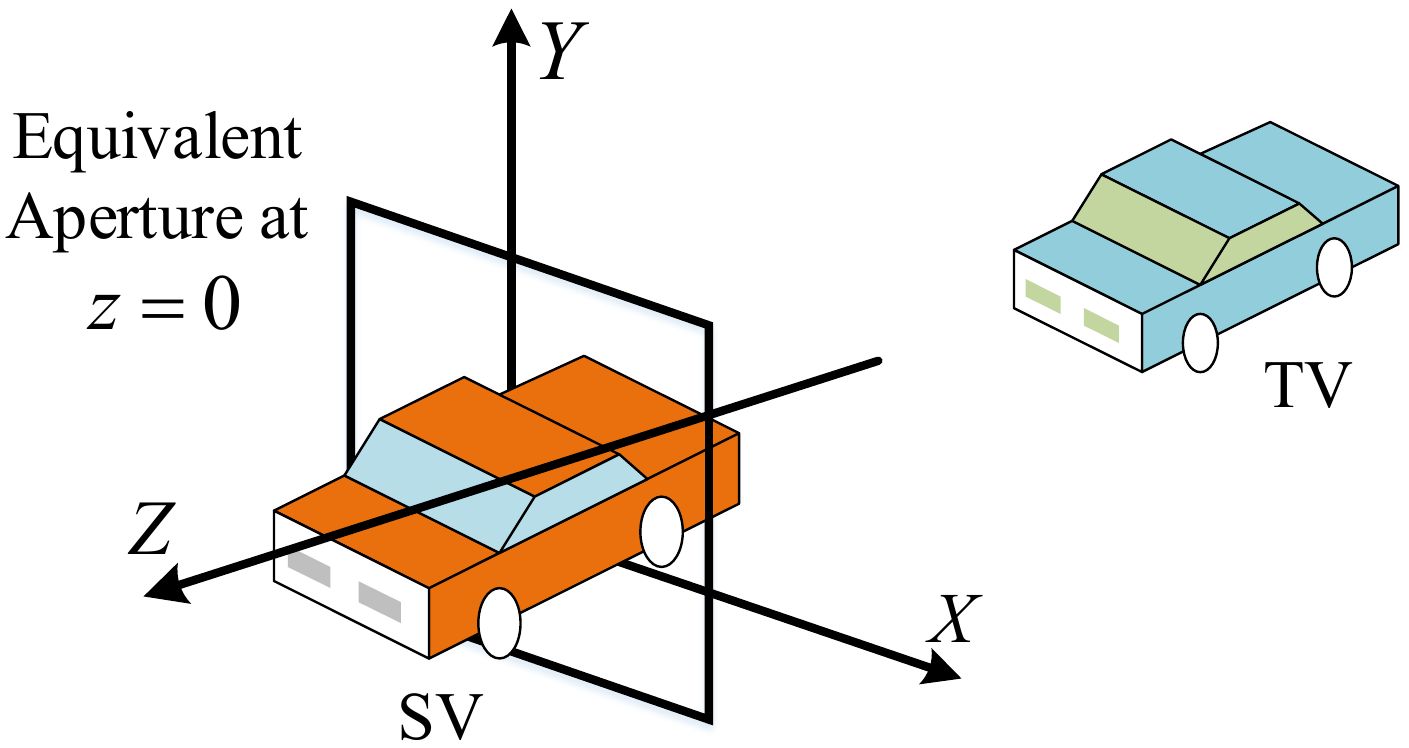}
		\caption*{(a) The equivalent aperture plane.}%加*可以去掉默认前缀，作为图片单独的说明
		\label{aperture}
	\end{minipage}
\hfill
	\begin{minipage}[t]{0.5\linewidth}%需要几张添加即可，注意设定合适的linewidth
		\centering
		\includegraphics[height = 90pt]{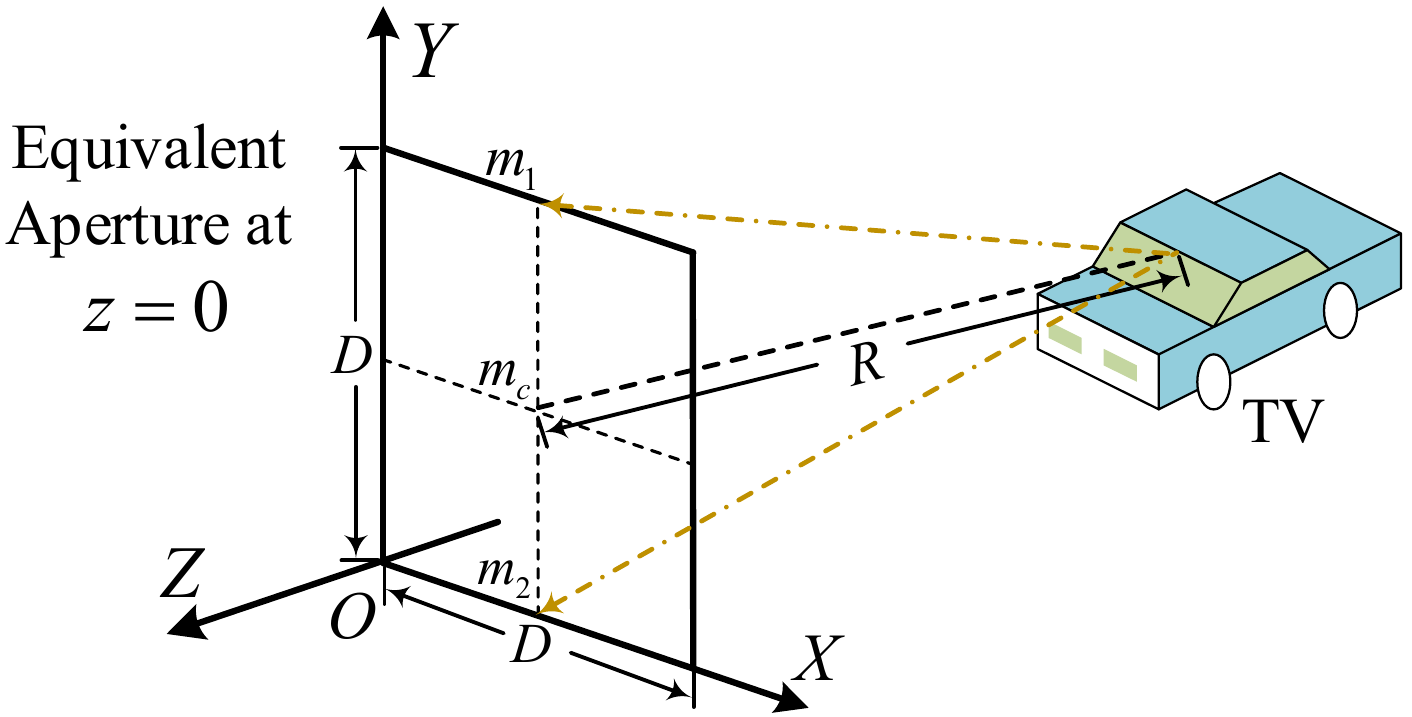}
		\caption*{(b) Illustration of azimuth and range resolutions.}
		\label{resolutionIllustration}
	\end{minipage}
	\caption{Illustration of azimuth and range resolutions.}\label{resolution}\vspace{-0.5cm}
\end{figure}
\begin{remark}[Sampling  Requirements] To achieve the above resolutions, there exist two kinds of sampling requirements on spatial and frequency domains.
	\begin{itemize}
		\item {\bf Spatial Sampling:} The spatial sampling represents the distance between two adjacent receive antennas. To achieve the resolution in \eqref{spatialResolution}, the receive antennas deployment over the ``equivalent aperture'' needs to meet the Nyquist sampling criterion such that the required sampling intervals are less than $\Delta_x$ and $\Delta_y$ to avoid aliasing. Therefore, the distances between two adjacent receive antennas along $X$ or $Y$ directions are
		\begin{align}\label{samplingReq}
		{\Delta _x} = {\Delta _y} \le \min_R \left\{\frac{c}{{2{f_c}}}\frac{{\sqrt {{4{ {{R}} }^2} + {D^2}} }}{2D}\right\}\mathop  = \limits^{(a)} \frac{c}{{4{f_c}}},
		\end{align}
		where (a) follows for the worst case with $R=0$. Moreover, since the $Z$-axis is defined as the direction of AoA which varies over time, the distance between two adjacent receive antennas in all directions should be smaller than $\frac{c}{{4{f_c}}}$.
		\item {\bf Frequency sampling:} {
			The frequency sampling interval refers to the frequency gap $\Delta$ to achieve the maximum ranging distance $R_{\max}$, given the resolution $\delta_z$ \eqref{RangeResolution}.
			Specifically, the number of the minimum samples for the resolution $\delta_z$ 
			is given as $\frac{R_{\max}}{\delta_z}$, which provides the following upper bound as
			\begin{align}\label{upperbound}
			\frac{R_{\max}}{\delta_z}\leq \frac{f_K-f_1}{\Delta}. 
			\end{align}
			Plugging \eqref{RangeResolution} into \eqref{upperbound} and with some manipulations, we have 
			\begin{align}
			\Delta\leq \frac{c}{R_{\max}}.
			\end{align}
			%which is equivalent to that of Remark \ref{feasibleDis}. 
		}
	\end{itemize}  
\end{remark}

%{In addition, we also give a brief resolution analysis on the repetitive single-point positioning method aforementioned in Sec. \ref{intro:sensing}. Supposing that the system bandwidth is $\mathcal{B}$ and the bandwidth allocated to each antenna is $\Delta_f \le \mathcal{B}$. %Then the range resolution is limited by $\delta_z = \frac{c}{\Delta_f}$. 
%Sequential detection in time achieves resolution $\delta_z = \frac{c}{B}$, while causing severe latency and detection error. Parallel single-point positioning cannot take the most of the system bandwidth by using orthogonal waveforms to avoid interference, resulting in resolution degradation, i.e. $\delta_z = \frac{c}{\Delta_f}$ where ${\Delta_f} \ll \mathcal{B}$.}

\begin{figure}[t]
	\begin{minipage}[t]{0.5\linewidth}			\centering
		\includegraphics[scale = 0.65]{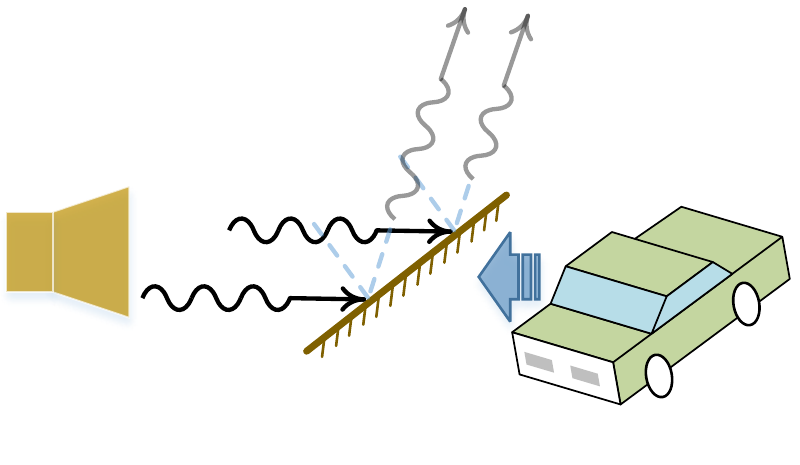}
		\caption*{(a) The effect of the TV's orientation.}			\label{fig:side:RCSph}
	\end{minipage}
\hfill
	\begin{minipage}[t]{0.5\linewidth}			\centering
		\includegraphics[scale = 0.45]{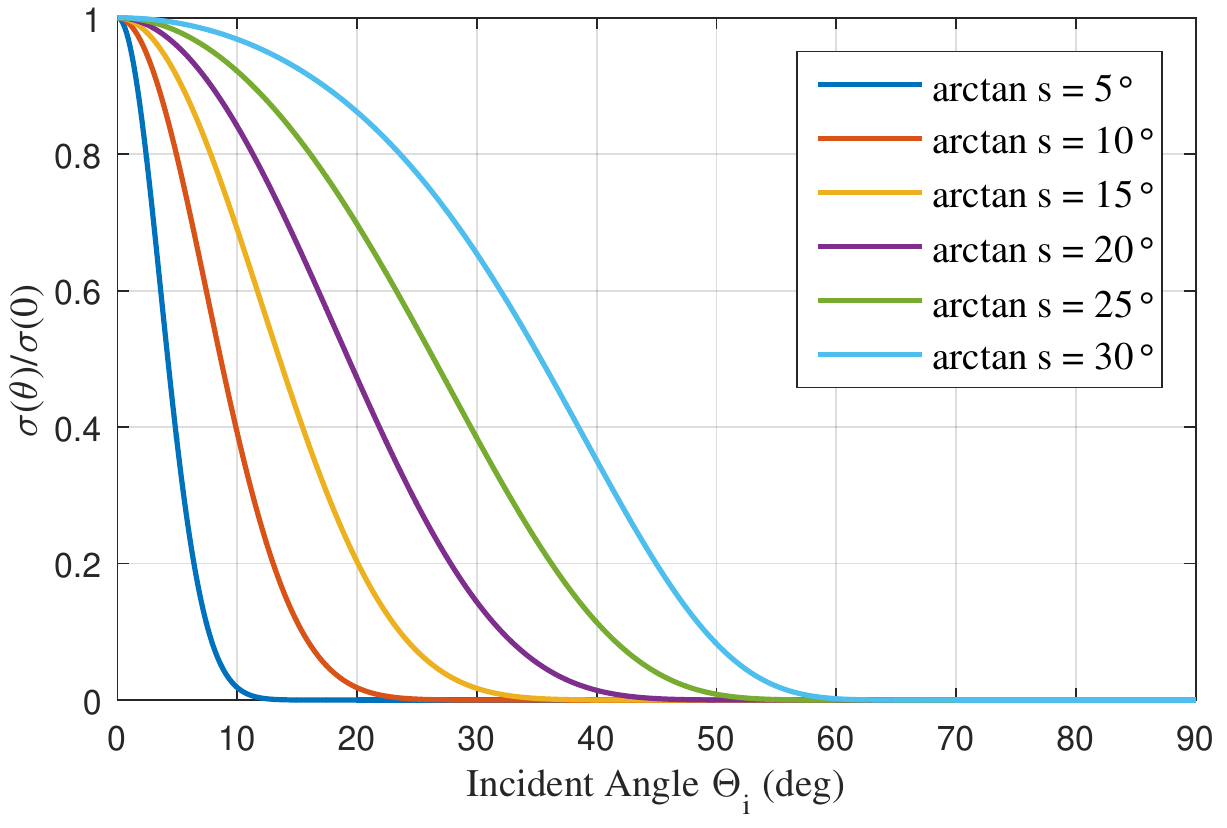}
		\caption*{(b) The impact of RCS versus the incident angle.}
		\label{fig:side:RCS}
	\end{minipage}
	\caption{The effect of RCS on positioning techniques based on RADAR systems.}\label{RCS_fig}\vspace{-0.5cm}
\end{figure}

\subsection{Propagation Loss Analysis}\label{Comparison_LoS}
The signal propagation loss determines the power level of the received signals, so as the performance of the detected position against noise. In this subsection, we analyze the propagation loss of the proposed COMPOP and existing multi-point positioning techniques, i.e. RADAR or LIDAR, over a LoS link for signal power comparison. It is obvious that the proposed technique experiences less power loss than RADAR-based techniques due to the half propagation distance. Besides, the orientation of the target surface 
is another factor to affect the reflected signal power as shown in Fig.~\ref{RCS_fig}(a). 
Specifically, for RADAR-based multi-point positioning techniques, the received signal power can be expressed as
\begin{align}\label{PowerLevel}
P_{r}^{(c)}={{P_{t}G_{t}} \over {4\pi R^{2}}}\sigma \left( \Theta_i  \right) {{1} \over {4\pi R^{2}}}A_{\mathrm {eff}} = {{P_{t}G_{t}\lambda^2} \over {64\pi^3 R^{4}}}\sigma \left( \Theta_i  \right),
\end{align}
where $P_t$ and $G_t$ are the input power and gain of the transmit antennas, and $R$ is the distance from the RADAR to the target. The other two factors, $\sigma$ and $A_{\mathrm {eff}} = \frac{\lambda^2}{4\pi}$ are \emph{radar cross-section} (RCS) of the target and the effective area of the RADAR receive antenna, respectively~\cite{RCSbook}. The RCS $\sigma$ 
%measures the effective area of the target that intercepts the transmitted RADAR power and then scatters that power isotropically back to the RADAR receiver, which is given as
is mathematically defined as
\begin{align}\label{RCS}
\sigma \left( \Theta_i  \right){\rm{ = }}\frac{{{{\left| {\Gamma_s(0)} \right|}^2}}}{{2{s^2}}}{\sec ^4}\Theta_i \exp \left\{ { - \frac{{{{\tan }^2}{\Theta_i}}}{{{s^2}}}} \right\},
\end{align}
where {$\Gamma_s(0)$ is the Fresnel reflection coefficient for normal incidence for each $s^2$, $\Theta_i$ is the incident angle} and $s^2$ is a parameter measuring the roughness of the target surface~\cite{RCSpaper}. The RCS is uncontrollable in the RADAR system design. From \eqref{PowerLevel} and \eqref{RCS}, the impact of the target's surface orientation, related to the incident angle $\Theta_i$, is shown clearly.  Fig.~\ref{RCS_fig}(b) indicates that the received signal power $P_r^{(c)}$ degrades seriously when the incident angle is large.
% Therefore, with complicated road condition and unpredictable TV orientation, conventional RADAR systems may fail to detect the object with certain probability. %Technique called bistatic RADAR is designed to handle this problem~\cite{bistatic}, while it may not work all the time and is difficult to be applied in the vehicle sensing scenario.

On the other hand, according to \eqref{PowerLevel}, the received signal power of our proposed technique can be expressed as $P_{r}={{P_{t}G_{t}\lambda^2} \over {{(4\pi R)}^{2}}}$, where the RCS has no influence on the received signal power of our proposed COMPOP since  the transmit antennas are isotropic. Moreover, reducing the signal propagation distance to half brings a power gain of $4\pi R^2$. Therefore, the proposed COMPOP is able to function well regardless the orientation of the TV and guarantee a signal power gain compared with conventional RADAR systems.

\section{COMPOP Using Surface Reflection in NLoS }\label{sec:NLoS}
Consider a NLoS scenario where the TV's LoS path to the SV is blocked but several NLoS paths reflected by nearby vehicles are available. We exploit mmWave spectrum's specular signal reflection on the smooth surfaces of nearby vehicles\footnote{MmWave spectrum has two reflection properties on vehicles' metal surfaces, where 1) its reflection coefficient is almost $1$ unless the incident angle is near $\frac{\pi}{2}$ or $0$, and 2) mmWave signals experience the mirror-like specular reflection \cite{Table60GHz}.}, establishing mirror-reflection links between the TV and the SV. Based on signals reflected from different paths, given as \eqref{demodulated}, multiple virtual vehicles symmetric to the actual TV concerning the reflection surfaces can be detected via LoS COMPOP. Since SV has no prior information of the reflection surfaces, an additional step should be required to combine multiple virtual TVs into the actual one.
To this end, we design a two-stage approach following 1) LoS COMPOP and 2) virtual TV combining, as illustrated in Fig. \ref{NLoS_Diag}.
%Moreover, we show that LoS case is a special realization of the NLoS case and can be solved in the same way. 	

{ Additional assumptions used in this section are summarized for clarification in the sequel:
\begin{itemize}
	\item {\bf Specular Reflection:} The surfaces of nearby vehicles or buildings are able to provide specular reflections for mmWave signals as mentioned above.
	\item {\bf Reflection Surfaces Vertical to the Ground:} We consider the case where the reflection surfaces are vertical to the ground, which is common in practice. This assumption can be easily relaxed as illustrated in Remark \ref{Vertical} later.
\end{itemize} }

\subsection{LoS COMPOP}
Firstly, we assume that the signals reflected from different paths can be perfectly resolved at the SV by the AoA differentiation. Then based on signals received from each path, the LoS COMPOP can be accomplished by the following two steps.
\subsubsection{Synchronization}\label{sec2:Syn}
The operation is similar to the synchronization under LoS in Sec. \ref{sec:Synchronization}. %except the detected locations of the two transmit antennas.
Consider the signals from path $\ell$ as an example.
Let $\sigma$ denote the system clock difference satisfying the following equations: 
\begin{align}\label{Sigma_NLoS}
	\sigma=\tau_{\mathsf{a},m}^{(\ell)}-\frac{\eta_{\mathsf{a},m}^{(\ell)}}{2\pi\Delta}=\tau_{\mathsf{b},m}^{(\ell)}-\frac{\eta_{\mathsf{b},m}^{(\ell)}}{2\pi\Delta},
\end{align}
where $\{\tau_{\mathsf{a},m}^{(\ell)}, \tau_{\mathsf{b},m}^{(\ell)}\}$ and $\{\eta_{\mathsf{a},m}^{(\ell)}, \eta_{\mathsf{b},m}^{(\ell)}\}$ represent the propagation time and phase difference of the signature waveforms from the antennas $\mathsf{a}$ and $\mathsf{b}$ through $\ell$-th mirror-reflection link, respectively. Following the procedures in Sec. \ref{sec:Synchronization} gives the locations of antennas $\mathsf{a}$ and $\mathsf{b}$ on the virtual TV $\ell$, denoted by $\mathbf{x}_{\mathsf{a}}^{(\ell)}$ and $\mathbf{x}_{\mathsf{b}}^{(\ell)}$, which are different from the real locations $\mathbf{x}_{\mathsf{a}}$ and $\mathbf{x}_{\mathsf{b}}$. As in the case with an LoS path, the points $\mathbf{x}_{\mathsf{a}}^{(\ell)}$ and $\mathbf{x}_{\mathsf{b}}^{(\ell)}$ are used to compensate the system clock difference, which enables virtual vehicle positioning in Sec. \ref{subsection:NLoSImage}. Besides, $\mathbf{x}_{\mathsf{a}}^{(\ell)}$ and $\mathbf{x}_{\mathsf{b}}^{(\ell)}$ help estimate the reflection surface elaborated in Sec. \ref{TVcombining}.

%Second, the system clock difference is estimated following the same algorithm for the case with an LoS path in Sec. \ref{sec:Synchronization}, which enables virtual vehicle positioning in Sec. \ref{subsection:NLoSImage}.  %It is worth noting that all virtual TVs share the same system clock difference since they are originated from the same TV. 

\begin{figure}[t]\vspace{-0.5cm}
	\centering
	\includegraphics[%height=240pt, 
	width=340pt]{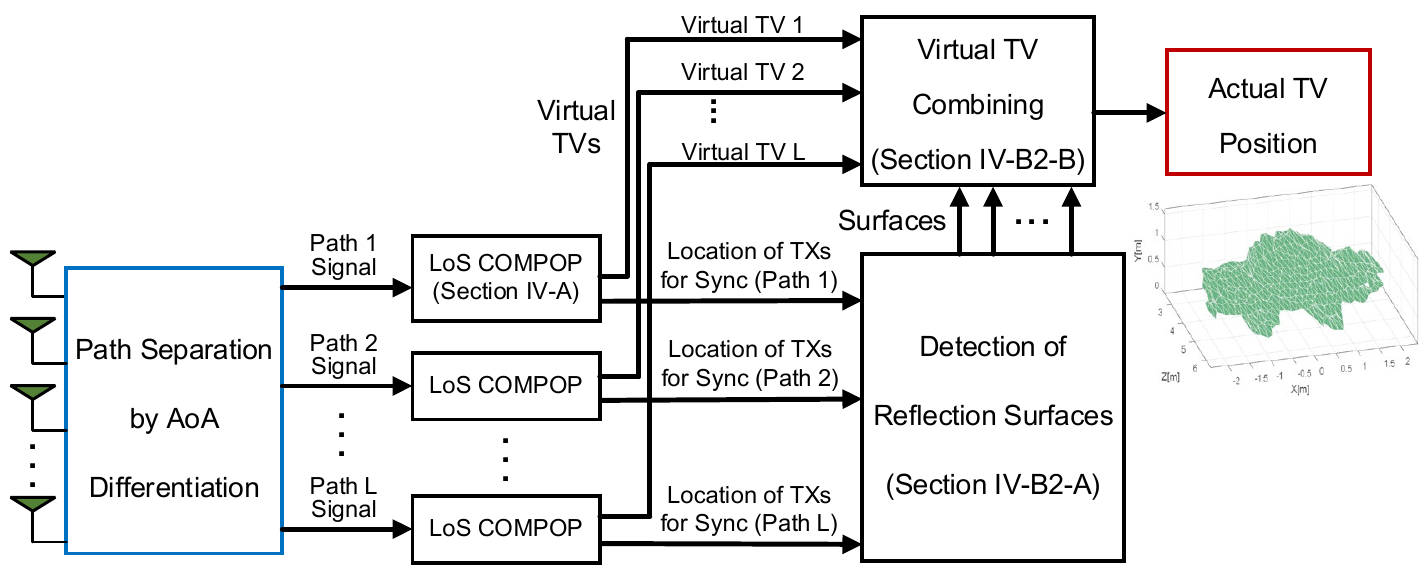}
	\caption{Diagram of COMPOP in NLoS.}\label{NLoS_Diag}\vspace{-0.5cm}
\end{figure}

%according to \eqref{Sigma_NLoS}, the system clock difference estimation is totally the same as that in Sec. \ref{sec:Synchronization}, and the same clock difference $\tilde \sigma$ can be calculated for signals from different reflection links but the same TV. The estimated system clock difference $\tilde \sigma$ will be used for the clock difference compensation in Sec. \ref{subsection:NLoSImage}, and $\mathbf{x}_{\mathsf{a}}^{(\ell)}$ and  $\mathbf{x}_{\mathsf{b}}^{(\ell)}$ will be used as input to detect the reflection surfaces in Sec. \ref{TVcombining}-A. 

\begin{remark}[TV Recognition in NLoS]
 Considering that multiple reflection links exist for the same TV in NLoS case, the detected system clock difference can help recognize the TVs. Specifically, signature waveforms transmitted from the same TV share the same system clock difference $\sigma$ regardless of different signal paths, since they are originated from the same TV.  Therefore, the SV can recognize signals from the same TV or different TVs according to the detected system clock difference $\tilde \sigma$.
\end{remark}
\subsubsection{Virtual TV Positioning}\label{subsection:NLoSImage}
Recall that plugging the perfectly estimated $\sigma$ into the demodulated signal \eqref{demodulated} gives the following  3D surface integral form of $y_m^{\ell,k}$ as
\begin{align}\label{NLosImageSig}
	{y}_m^{\ell,k} = {{\mathsf{y}}}^{(\ell)}(\mathbf{p}_m, f_k) &%\nonumber \\
	=  \Gamma^{(\ell)}\int_{\mathbb{R}^3} {\mathsf{I}_{\{\mathbf{x}^{(\ell)} \in \mathcal{X}^{(\ell)}\}}{\exp\left( - j\frac{2\pi f_k}{c} {D^{(\ell)}{(\mathbf{x}, \mathbf{p}_{m})}}\right)}d{\mathbf{x}^{(\ell)}}}, 
\end{align}
where ${{\mathsf{y}}}^{(\ell)}(\mathbf{x}, f): \mathbb{R}^4\to\mathbb{R}$ is a continuous function, and ${D^{(\ell)}{(\mathbf{x}, \mathbf{p}_{m})}}$ represents the Euclidean distance from point $\mathbf{x}^{(\ell)}$ in $\mathcal{X}^{(\ell)}$, which is symmetric to point $\mathbf{x}$ in $\mathcal{X}$ w.r.t. the $\ell$-th reflection surface, to the location of the SV's antenna $m$ denoted by $\mathbf{p}_m=(p_m^{(x)}, p_m^{(y)}, p_m^{(z)})^T$. 
{ Compared with \eqref{imageSignal} in the LoS scenario,  
the reflection coefficient $\Gamma^{(\ell)}$ can be considered as a constant scaling factor, which does not affect the positioning procedure. } Following the same steps in Sec. \ref{sec:Imaging},  
the virtual TV $\ell$'s power spectrum $\Phi[\mathbf{x}^{(\ell)}]$ can be calculated. After the peak detection in \eqref{Mapping_Rule}, we can obtain the virtual locations of the transmit antennas $\mathcal{X}^{(\ell)}$, which is referred to as virtual TV $\ell$. The estimated virtual TVs are used to position the actual TV explained in Sec. \ref{TVcombining}.

\subsection{Virtual TV Combining}\label{sec2:Common-point}
\subsubsection{Overview}
This subsection aims at positioning the actual TV with multiple virtual TVs detected in Sec. \ref{subsection:NLoSImage} under the assumption that the reflection surfaces are vertical to the ground, which is common in practice.
%To this end, we use the location of transmit antennas $\mathsf a$ and $\mathsf b$ detected in Sec. \ref{sec2:Syn} as the representative points,  whose coordinates are $\mathbf{x}^{(\ell)}_{\mathsf{a}}=\left(x^{(\ell)}_{\mathsf{a}},y^{(\ell)}_{\mathsf{a}},z^{(\ell)}_{\mathsf{a}}\right)$ and $\mathbf{x}^{(\ell)}_{\mathsf{b}}=\left(x_{\mathsf{b}}^{(\ell)},y_b^{(\ell)},z_{\mathsf{b}}^{(\ell)}\right)$, which have geographical relations with the counterpart points on the actual TV denoted by $\mathbf{x}_{\mathsf{a}}=\left(x_{\mathsf{a}},y_{\mathsf{a}},z_{\mathsf{a}}\right)$ and $\mathbf{x}_{\mathsf{b}}=\left(x_{\mathsf{b}},y_{\mathsf{b}},z_{\mathsf{b}}\right)$  summarized in Lemma \ref{CommonPointRelation}. 
Since the actual TV can be directly obtained by shifting the virtual TVs w.r.t. their corresponding reflection surfaces, the key step becomes the detection of the reflection surfaces. 
%Moreover, it can be inferred from Fig. \ref{location} that the reflection surfaces can be straightforward obtained after the detection of $(\mathbf{x}_{\mathsf{a}}, \mathbf{x}_{\mathsf{b}})$. 
To this end, an algorithm to position the actual TV, consists of two steps: 1) detection of reflection surfaces; and 2) combining virtual TVs. The detailed procedures are illustrated in the following.

%Some intuitions are made from Lemma \ref{CommonPointRelation}. First, it is shown in \eqref{representXA} that the coordinates of  $\left(\mathbf{x}_{\mathsf{a}},\mathbf{x}_{\mathsf{b}}\right)$ can be calculated when $\theta_{\ell_1}$ and $\theta_{\ell_2}$ are given. Second, noting that $\phi_{\ell_1}$ and $\phi_{\ell_2}$ are observable from virtual TVs $\ell_1$ and $\ell_2$ directly, 
%$\theta_{\ell_2}$ is easily calculated when $\theta_{\ell_1}$ is given. Last, $\theta_{\ell_1}$ and the resultant $\left(\mathbf{x}_{\mathsf{a}},\mathbf{x}_{\mathsf{b}}\right)$ are said to be correct if another combination of two virtual TVs, for example $\ell_1$ and $\ell_3$, can yield the equivalent result of $\left(\mathbf{x}_{\mathsf{a}},\mathbf{x}_{\mathsf{b}}\right)$. 
%As a result, we can lead to the following feasibility condition. 

\begin{figure}[t]\vspace{-0.5cm}
	\centering
	\includegraphics[scale = 0.48]{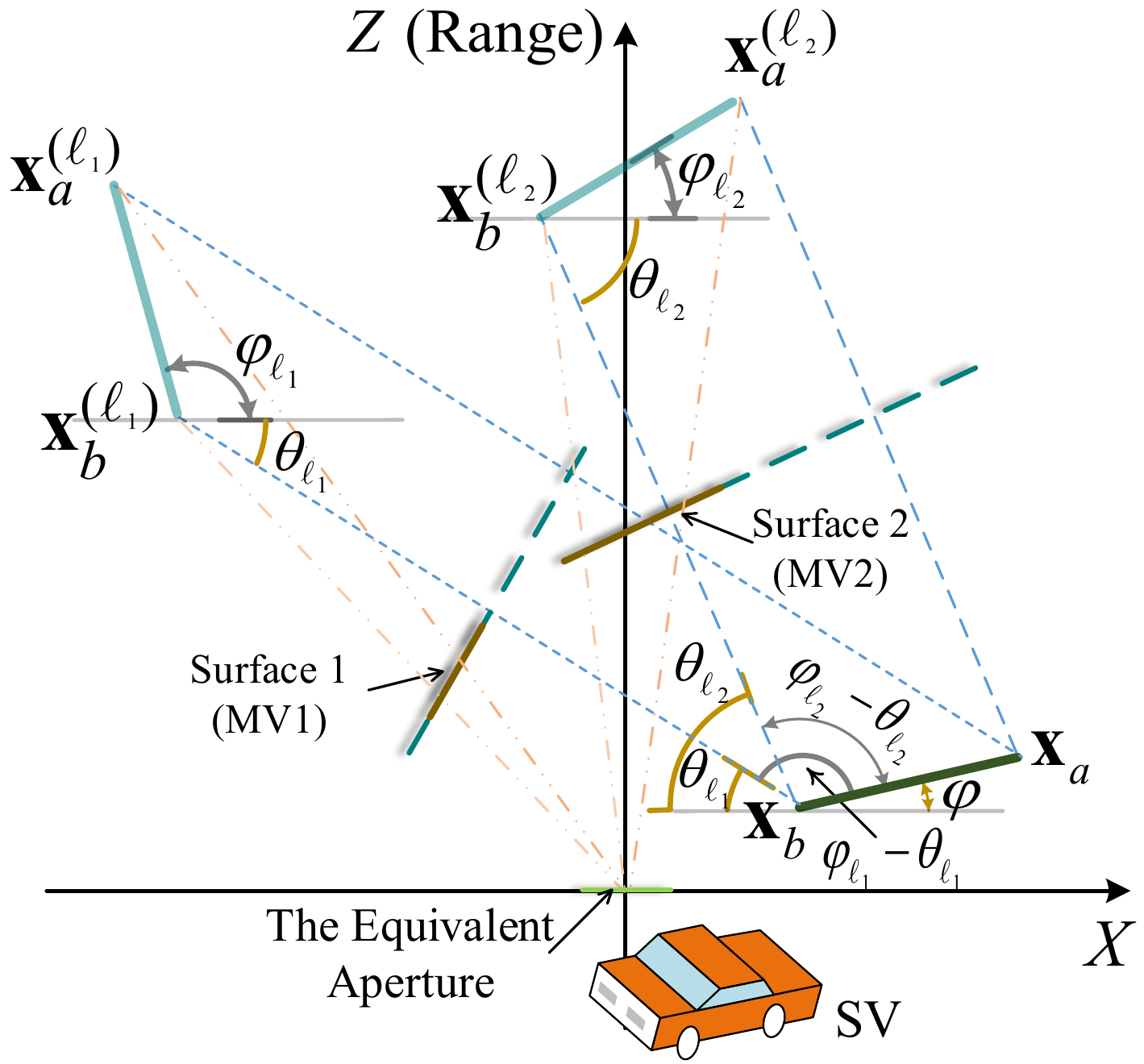}
	\caption{Geometric relations between the virtual and actual TVs from the top view.}\label{location}\vspace{-0.5cm}
\end{figure}

\subsubsection{Algorithm Description}\label{TVcombining}
\begin{itemize}
	\item \emph{A. Detection of Reflection Surfaces:} Consider the virtual locations of representative transmit antennas $\mathsf a$ and $\mathsf b$ detected in Sec. \ref{sec2:Syn} with coordinates $\mathbf{x}^{(\ell)}_{\mathsf{a}}=\left(x^{(\ell)}_{\mathsf{a}},y^{(\ell)}_{\mathsf{a}},z^{(\ell)}_{\mathsf{a}}\right)^T$ and $\mathbf{x}^{(\ell)}_{\mathsf{b}}=\left(x_{\mathsf{b}}^{(\ell)},y_b^{(\ell)},z_{\mathsf{b}}^{(\ell)}\right)^T$, which are symmetric to the counterpart points on the actual TV denoted by $\mathbf{x}_{\mathsf{a}}=\left(x_{\mathsf{a}},y_{\mathsf{a}},z_{\mathsf{a}}\right)^T$ and $\mathbf{x}_{\mathsf{b}}=\left(x_{\mathsf{b}},y_{\mathsf{b}},z_{\mathsf{b}}\right)^T$ w.r.t. the reflection surface $\ell$. It can be inferred from Fig. \ref{location} that given the locations $(\mathbf{x}^{(\ell)}_{\mathsf{a}},  \mathbf{x}^{(\ell)}_{\mathsf{b}})$, the reflection surfaces can be straightforwardly obtained if $(\mathbf{x}_{\mathsf{a}}, \mathbf{x}_{\mathsf{b}})$ is estimated. Therefore, the problem is translated into the detection of representative transmit antennas $(\mathbf{x}_{\mathsf{a}}, \mathbf{x}_{\mathsf{b}})$. To this end, we summarize the geometric relation between $(\mathbf{x}^{(\ell)}_{\mathsf{a}},  \mathbf{x}^{(\ell)}_{\mathsf{b}})$ and $(\mathbf{x}_{\mathsf{a}}, \mathbf{x}_{\mathsf{b}})$ in the following lemma.  
	
	\begin{lemma}\label{CommonPointRelation} {Consider virtual TVs $\ell_1$ and $\ell_2$ whose representative points are $\{\mathbf{x}_{\mathsf{a}}^{(\ell_1)}, \mathbf{x}_{\mathsf{b}}^{(\ell_1)}\}$ and $\{\mathbf{x}_{\mathsf{a}}^{(\ell_2)}, \mathbf{x}_{\mathsf{b}}^{(\ell_2)}\}$ respectively, 
			which have geometric relations with the counterpart points on the actual TV denoted by $\left(\mathbf{x}_{\mathsf{a}},\mathbf{x}_{\mathsf{b}}\right)$ as follows.
			\begin{enumerate}
				\item Let $\theta_{\ell}$ denote the directed angle from the $X$-axis (parallel to the ground) to the line between $\mathbf{x}_{\mathsf{a}}^{(\ell)}$ and $\mathbf{x}_{\mathsf{a}}$ or $\mathbf{x}_{\mathsf{b}}^{(\ell)}$ and $\mathbf{x}_{\mathsf{b}}$ (see Fig. \ref{location}). { According to the geometric relation illustrated in Fig. \ref{location}, $\mathbf{x}_{\mathsf{a}}$ can be given in terms of $\theta_{\ell_1}$ and $\theta_{\ell_2}$ as
				\begin{align}\label{representXA}
				\mathbf{x}_{\mathsf{a}}=\left(\begin{aligned}
				& x_{\mathsf{a}}\\
				& y_{\mathsf{a}}\\
				& z_{\mathsf{a}}
				\end{aligned}\right)=\left(
				\begin{aligned}
				&	\frac{{\left( {{z_{\mathsf{a}}^{(\ell_1)}} - {z_{\mathsf{a}}^{(\ell_2)}}} \right) + \left( {{x_{\mathsf{a}}^{(\ell_2)}}\tan {(\theta_{\ell_2})} - {x_{\mathsf{a}}^{(\ell_1)}}\tan {(\theta_{\ell_1})}} \right)}}{{\tan (\theta_{\ell_2}) - \tan (\theta_{\ell_1})}} \\
				&      \textrm{$y_{\mathsf{a}}^{(\ell_1)}$ or $y_{\mathsf{a}}^{(\ell_2)}$}  \\
				&	\textrm{${z_{\mathsf{a}}^{(\ell_1)}} + \tan \left({\theta_{\ell_1}}\right)\left( {x_{\mathsf{a}} - {x_{\mathsf{a}}^{(\ell_1)}}} \right)$}
				\end{aligned}
				\right), 
				\end{align}
				and  similarly the expression of $\mathbf{x}_{\mathsf{b}}$ in terms of $\theta_{\ell_1}$ and $\theta_{\ell_2}$ is obtained by replacing all $\mathsf{a}$ in \eqref{representXA} with $\mathsf{b}$. }
				\item Let $\phi_{\ell}$ denote the directed angle from the $X$-axis to the line segment of virtual TV $\ell$, from $\mathbf{x}^{(\ell)}_{\mathsf{a}}$ to $ \mathbf{x}^{(\ell)}_{\mathsf{b}}$, as shown in Fig. \ref{location}. The angles $\theta_{\ell}$ and $\phi_{\ell}$ of two virtual TVs $\ell_1$ and $\ell_2$ follows the relation
				\begin{align}\label{angleRelation}
				\theta_{\ell_1}-\theta_{\ell_2}=\frac{\phi_{\ell_1}-\phi_{\ell_2}}{2}. 
				\end{align}
			\end{enumerate}
		}
	\end{lemma}
	\begin{proof}
		Please refer to Appendix  C. 
	\end{proof}
	
	%\textcolor{blue}{Moreover, by plugging one pair of $(\theta_{1}, \theta_{\ell})$ into \eqref{representXA}, we can calculate an estimation on the locations $\mathsf{a}$ and $\mathsf{b}$, denoted as $({\mathbf{z}}_{\mathsf{a}}^{(\ell)}, {\mathbf{z}}_{\mathsf{b}}^{(\ell)})$.} 

	Based on Lemma \ref{CommonPointRelation}, the representative points $\left(\mathbf{x}_{\mathsf{a}},\mathbf{x}_{\mathsf{b}}\right)$ are estimated as follows.
	First, all angles $\{\theta_{\ell}\}_{\ell=2}^L$ can be expressed in terms of $\theta_{1}$ as $\theta_{\ell}=\theta_{1}+\frac{\phi_{\ell}-\phi_{1}}{2}$ using 
 \eqref{angleRelation}. Next, plugging each pair of $(\theta_{1}, \theta_{\ell})$ into \eqref{representXA} makes it possible to express the locations $\mathsf{a}$ and $\mathsf{b}$ w.r.t $\theta_1$, denoted by $({{\mathbf{v}}}_{\mathsf{a}}^{(\ell)}(\theta_1), {\mathbf{v}}^{(\ell)}_{\mathsf{b}}(\theta_1))$. If $\theta_{1}$ is correct, $({{\mathbf{v}}}_{\mathsf{a}}^{(\ell)}(\theta_1), {\mathbf{v}}^{(\ell)}_{\mathsf{b}}(\theta_1))$ naturally coincides with $\left(\mathbf{x}_{\mathsf{a}},\mathbf{x}_{\mathsf{b}}\right)$. In other words, estimating $\left(\mathbf{x}_{\mathsf{a}},\mathbf{x}_{\mathsf{b}}\right)$ is translated into finding $\theta_{1}$ minimizing the following squared Euclidean distance as
	\begin{align}\label{OptTheta}
	\theta_{1}^*=\arg\min_{\theta_{1}}\sum\limits_{\ell_1=2}^{L}\sum\limits_{\ell_2=2}^{L} {\left(\left\| {\mathbf{v}}_{\mathsf{a}}^{(\ell_1)}(\theta_1) - {\mathbf{v}}_{\mathsf{a}}^{(\ell_2)} (\theta_1)\right\|_2+\left\| {\mathbf{v}}_{\mathsf{b}}^{(\ell_1)}(\theta_1) - {\mathbf{v}}_{\mathsf{b}}^{(\ell_2)}(\theta_1) \right\|_2\right) }.
	\end{align}
	The optimal $\theta_{1}^*$ is computed by 1D search over $ \left[ { - {\pi },{\pi }} \right]$, and the resultant ${\mathbf{x}}_{\mathsf{a}}^*=\frac{1}{L-1}\sum_{\ell=2}^L{\mathbf{v}}_{\mathsf{a}}^{(\ell)}(\theta_{1}^*)$ and 
	${\mathbf{x}}_{\mathsf{b}}^*=\frac{1}{L-1}\sum_{\ell=2}^L{\mathbf{v}}_{\mathsf{b}}^{(\ell)}(\theta_{1}^*)$ can be directly obtained by using the optimal $\theta_{1}^*$. %Note that in the case without phase error, the optimal solution of the problem \ref{OptTheta} is zero and $({\mathbf{x}}_{\mathsf{a}}^*, {\mathbf{x}}_{\mathsf{b}}^*)=\left(\mathbf{x}_{\mathsf{a}},\mathbf{x}_{\mathsf{b}}\right)$.	
	Then with $({\mathbf{x}}_{\mathsf{a}}^*, {\mathbf{x}}_{\mathsf{b}}^*)$, the reflection surface $\ell$, which is located on the middle of ${\mathbf{x}}_{\mathsf{a}}^*$ and ${\mathbf{x}}_{\mathsf{a}}^{(\ell)}$, can be expressed by a line because it is perpendicular to $X-Z$ plane such that
	\begin{align}\label{reflectionPlane}
	z = - \frac{1}{{\tan ({\theta_{\ell}^*})}}\big(x- \frac{{x_{\mathsf{a}}^{(\ell)}} + {x_{\mathsf{a}}^*}}{2} \big)+\frac{\big( {{z_{\mathsf{a}}^{(\ell)}}+{z_{\mathsf{a}}^*}} \big)}{2},
	\end{align}
	where $\mathbf{x}_{\mathsf{a}}^*={({x}_{\mathsf{a}}^*, {y}_{\mathsf{a}}^*, {z}_{\mathsf{a}}^*)}^T$, and $\theta_{\ell}^*=\theta_{1}^*+\frac{\phi_{\ell}-\phi_{1}}{2}$. 
	
	\begin{prop}[Feasibility Condition for NLoS Position Combining]\label{Prop2}\emph{To detect the actual TV in NLoS, at least three reflection surfaces are required: $L\geq3$.}
	\end{prop}
	\begin{proof}
		Please refer to Appendix D. 
	\end{proof}
	%It is worth noting that the  line \eqref{reflectionPlane} passes the middle points of all lines between the real points $\mathbf{x}$ and virtual TVs $\mathbf{x}^{(\ell)}$, leading to the following mapping rule.
	\item \emph{B. Combining Virtual TVs:} With the knowledge of the reflection surfaces, the actual TV $\mathcal{X}$ can be obtained as elaborated in the following proposition.
	{
	\begin{prop}[Position Combining]\label{Prop3} \emph{Consider the virtual TV $\ell$ represented by $\mathcal{X}^{(\ell)}$. Given $\theta_{\ell}^*$ and ${\mathbf{x}}_{\mathsf{a}}^*$ (or ${\mathbf{x}}_{\mathsf{b}}^*$), the actual TV $\mathcal{X}$ can be obtained by the following mapping function:
			\begin{align}\label{FindLocation}
			\mathcal{X}={\mathsf{G}(\mathcal{X}^{(\ell)})}
			\end{align}
			where
			\begin{align}
			\mathsf{G}(\mathbf{x}^{(\ell)})&=\left(
			x^*, y^*, z^* \right)^T\nonumber\\
			&=\left(\!\begin{aligned}
			& {x^{(\ell )}} \!+\! \frac{{\tan (\theta _\ell ^*)}}{1+{\tan^2 (\theta _\ell ^*)} }\!\bigg(\! {\frac{{x_{\mathsf{a}}^{(\ell )} \!+\! x_{\mathsf{a}}^*}}{{\tan (\theta _\ell ^*)}} \!+\! z_{\mathsf{a}}^{(\ell )} \!+\! z_{\mathsf{a}}^* \!-\! \frac{{2{x^{(\ell )}}}}{{\tan (\theta _\ell ^*)}} \!-\! 2{z^{(\ell)}}} \!\bigg)\\
			& y^{(\ell)}\\
			& {z^{(\ell )}} + \frac{{\tan^2 (\theta _\ell ^*)}}{1+{\tan^2 (\theta _\ell ^*)} }\bigg( {\frac{{x_{\mathsf{a}}^{(\ell )} + x_{\mathsf{a}}^*}}{{\tan (\theta _\ell ^*)}} + z_{\mathsf{a}}^{(\ell )} + z_{\mathsf{a}}^* - \frac{{2{x^{(\ell )}}}}{{\tan (\theta _\ell ^*)}} - 2{z^{(\ell )}}} \bigg)
			\end{aligned}\!\right), \ \mathbf{x}^{(\ell)} \!\in\! \mathbb{R}^3\nonumber
			\end{align}
		}
	\end{prop}
	\begin{proof}
		Please refer to Appendix E. 
	\end{proof} }
\end{itemize}

%\begin{remark}[Selection of Virtual Position] It is important to select one virtual position to transform the RP since each virtual position experiences different levels of the localization errors on its propagation distance and angles. We select the virtual position of which the estimated propagation distance is shortest, which is verified by simulation to give the most accurate RP among some simple criteria.    
%\end{remark}
{
\begin{remark}[Existence of LoS path]{The LoS case is a special realization of the NLoS case, where one couple of representative points $\left(\mathbf{x}_{\mathsf{a}}^{\text{LoS}},\mathbf{x}_{\mathsf{b}}^{\text{LoS}}\right)$ are equivalent to the exact location $\left(\mathbf{x}_{\mathsf{a}},\mathbf{x}_{\mathsf{b}}\right)$. Therefore, all mathematical expressions for the virtual TV combining still hold when the LoS link exists, and the resultant $\left(\mathbf{x}_{\mathsf{a}}^*,\mathbf{x}_{\mathsf{b}}^*\right)$ can be obtained in the same way. }	
\end{remark} }
{
\begin{remark}[Arbitrary Reflection Surfaces]\label{Vertical}
The assumption that all reflection surfaces are vertical to the ground can be easily relaxed with one more antenna transmitting signature waveform $s_\mathsf{c}$. In this case, representative points $\{\mathbf{x}_\mathsf{a}^{\ell}, \mathbf{x}_\mathsf{b}^{\ell}, \mathbf{x}_\mathsf{c}^{\ell} | \forall \ell\}$ are detected in the clock synchronization step. Analogous to the result in \eqref{reflectionPlane}, reflection surfaces can then be located, without the assumption, by exploiting the information of these representative points, so as the actual TV position. %The detailed illustration is omitted here due to page limitation.
\end{remark}
}

\subsection{Propagation Loss Analysis}\label{Comparison_NLoS} 
Similar to the analysis in Sec. \ref{Comparison_LoS}, we also give signal power comparison between the proposed COMPOP and traditional RADAR techniques over the reflection links. Although RADAR techniques do not provide solutions to positioning without a LoS, %whose effective received signal power can thereby consider as $0$. In spite of this,
we still provide the received signal power of a reflection link, say the $\ell$-th reflection link, as
\begin{align}\label{PowerLevel_NLoS_RADAR}
	P_{r}^{(c)}={{P_{t}G_{t}\lambda^2} \over {4^5\pi^5 (R_1^{(\ell)})^{4} (R_2^{(\ell)})^4}} \sigma ( \Theta_i ) \sigma^2 ( \Theta_i^{(\ell)}),
\end{align}
where $\Theta_i^{(\ell)}$ is the incident angle on the reflection surface $\ell$, $R_1^{(\ell)}$ is the distance from the SV to the reflection surface $\ell$, and $R_2^{(\ell)}$ is the distance from the TV to the reflection surface $\ell$.
For the proposed COMPOP, signals from each reflection link give one virtual TV detected with algorithm in Sec. \ref{subsection:NLoSImage}. For the reflection link $\ell$, the received signal power is 
\begin{align}\label{PowerLevel_NLoS}
	P_{r}^{(\ell)}={{P_{t}G_{t}\lambda^2} \over {64\pi^3 (R_1^{(\ell)})^{2} (R_2^{(\ell)})^2}}\sigma ( \Theta_i^{(\ell)}).
\end{align}
It is obvious that the proposed COMPOP achieves a much lower propagation loss than the conventional RADAR techniques.

\section{Simulation Results}

In this section, the performance of the proposed COMPOP technique is evaluated by realistic settings. 
Signature waveforms at $4$ frequencies are used for the synchronization procedure. 
The number of frequencies used in the SFCW \eqref{transmitSig} is $K=256$ ranging from $57$ GHz to $60$ GHz with the constant gap $\Delta = 11.72$ MHz. The four frequencies used in the signature waveforms  \eqref{Signature_Waveform} are $(57-\Delta \cdot i)$ GHz where $i=1,2,3,4$. 
The numbers of the TV's and SV's antennas are ${N_t} = {N_r} = 200$, which are uniformly deployed on the vehicles bodies. The size of the equivalent receive aperture is $1 \times 1$ $m^2$.  The SNR of each received signal is fixed to $10$ dB. The number of reflection surfaces are $3$ and the distance between TV and SV is $8$ m unless stated otherwise. 

For the performance metric, we use the Hausdorff distances defined as follows. 
\begin{defn}[Hausdorff distance~\cite{Hausdorff}] The Hausdorff distance is widely used to evaluate the image retrieval performance by measuring  the similarity between two images. Consider an image $\mathcal{A}$ and its retrieved one $\mathcal{B}$, both of which are composed of discrete points. The Hausdorff distance is defined as 
	\begin{align}\label{Hausdorff_Def}
	\mathsf{H}\left( {\mathcal{A},\mathcal{B}} \right) = \max \left( {\mathsf{h}\left( {\mathcal{A},\mathcal{B}} \right), \mathsf{h}\left( {\mathcal{B},\mathcal{A}} \right)} \right),
	\end{align} 
	where $\mathsf{h}\left( {\mathcal{A},\mathcal{B}} \right) = \mathop {\max }\limits_{a \in \mathcal{A}} \mathop {\min }\limits_{b \in \mathcal{B}} \left\| {a - b} \right\|$. 
\end{defn}

\subsection{Graphical Example of Multi-Point Positioning}

This subsection aims at explaining the entire vehicular positioning procedure with step-by-step graphical examples. To this end, 
we consider the topology with three reflection surfaces illustrated in Fig. \ref{InitialSetting}. 
%The initial topology is illustrated in Fig.~\ref{InitialSetting}.  
%The $X$, $Y$, $Z$-axes are the SV's moving direction, height and depth, respectively.
The TV is represented by discrete points, each of which is  one TV's  antenna.  The size of the TV is $3\times1\times0.6$ $m^3$.
The equivalent receive aperture is parallel to $Y$-axis and located at $(0, 0, 0)$.
The equations of three reflection surfaces are given as $z = 1.02x + 3$, $z = \frac{{x + 13}}{4}$, and $z = 3x + 4$. 
\begin{figure}[t]\vspace{-0.5cm}
	\centering
	\includegraphics[scale = 0.35]{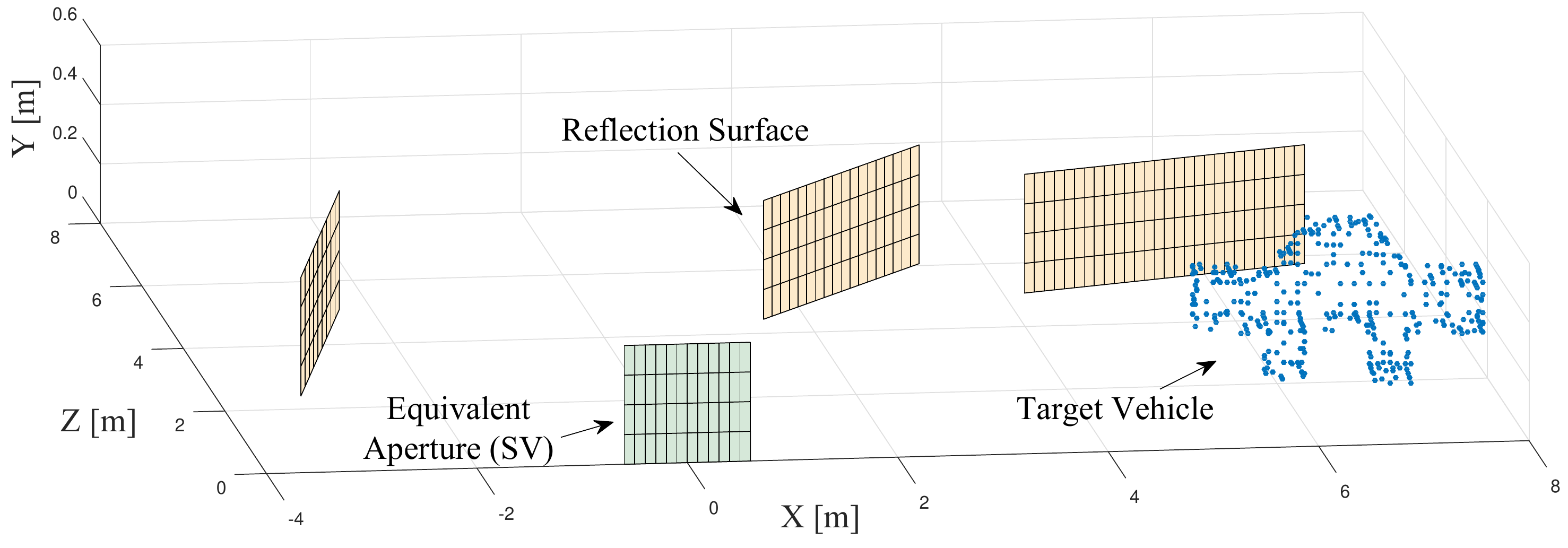}
	\caption{Initial topology.}\label{InitialSetting}\vspace{-0.5cm}
\end{figure}
\begin{figure}[t]
	\centering
	\includegraphics[scale = 0.35]{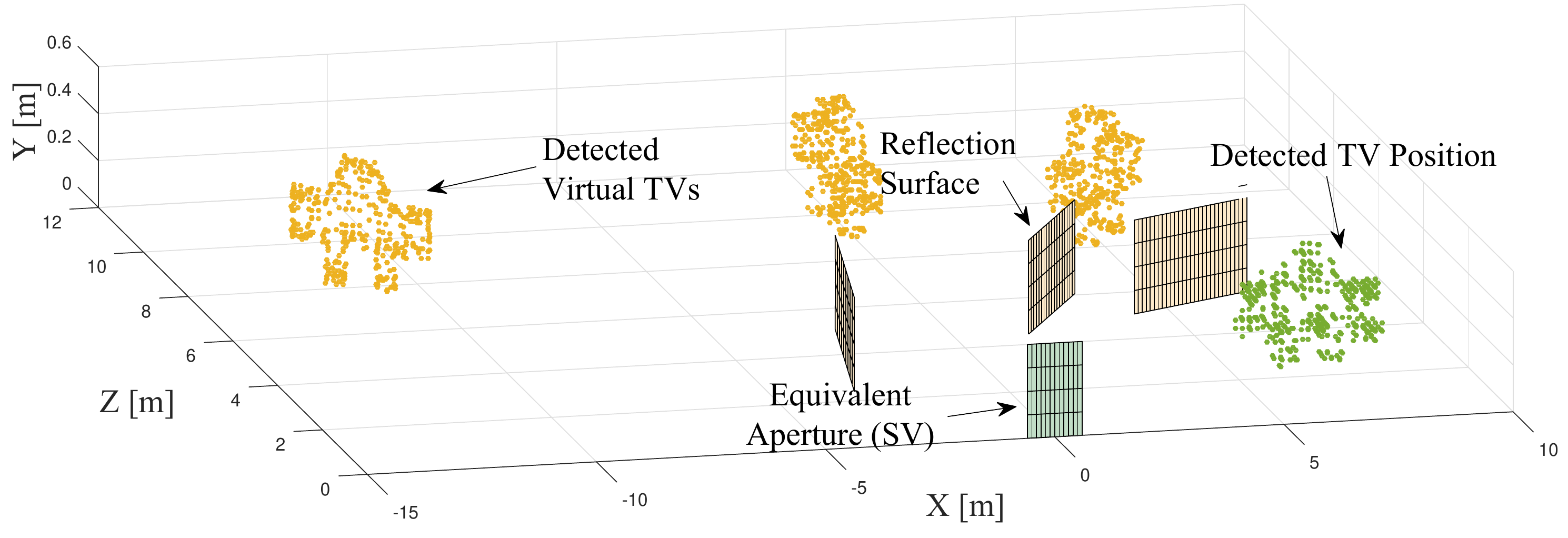}
	\caption{The detection of the virtual TVs.}\label{Reconstruct1st}\vspace{-0.5cm}
\end{figure}

Using the reflection surfaces, the SV detects three virtual TVs represented by yellow slots in Fig. \ref{Reconstruct1st}, each of which is  differentiable using AoA information. By the intelligent combining algorithm in Proposition \ref{Prop3}, each virtual TV can be shifted to its real location represented by green spots, of which the Hausdorff distance is $0.355$m, which is relatively small compared to the size of the TV. After graphical rendering process, the final detected position is obtained as in Fig. \ref{EnvelopCompare}(b) that is similar to the original one in Fig. \ref{EnvelopCompare}(a).

% The Hausdorff distance between the reconstructed image $\mathcal{A}$ and the ideal TV model $\mathcal{B}$ is $\mathsf{H}(A,B)=0.1262m$, which is relatively small compared to the size of the TV model.

%	Without loss of generality, signals received from the direction of the 1-st VP is selected. Therefore, the SV can receive the signals reflected by the 1-st MV and distinguish the signature waveforms from the common points in the 1-st VP. Then the synchronization procedures are taken and the location of the common points can be figured out. After synchronization, the SV can reconstruct the 1-st VP based on the received signals. The performance of the VPs is shown in Fig.~\ref{Reconstruct1st}. The reflection surface and the receive aperture in SV are also plotted.

\begin{figure}[t]
	\begin{minipage}[t]{0.5\linewidth}%设定图片下字的宽度，在此基础尽量满足图片的长宽
		\centering
		\includegraphics[scale = 0.35]{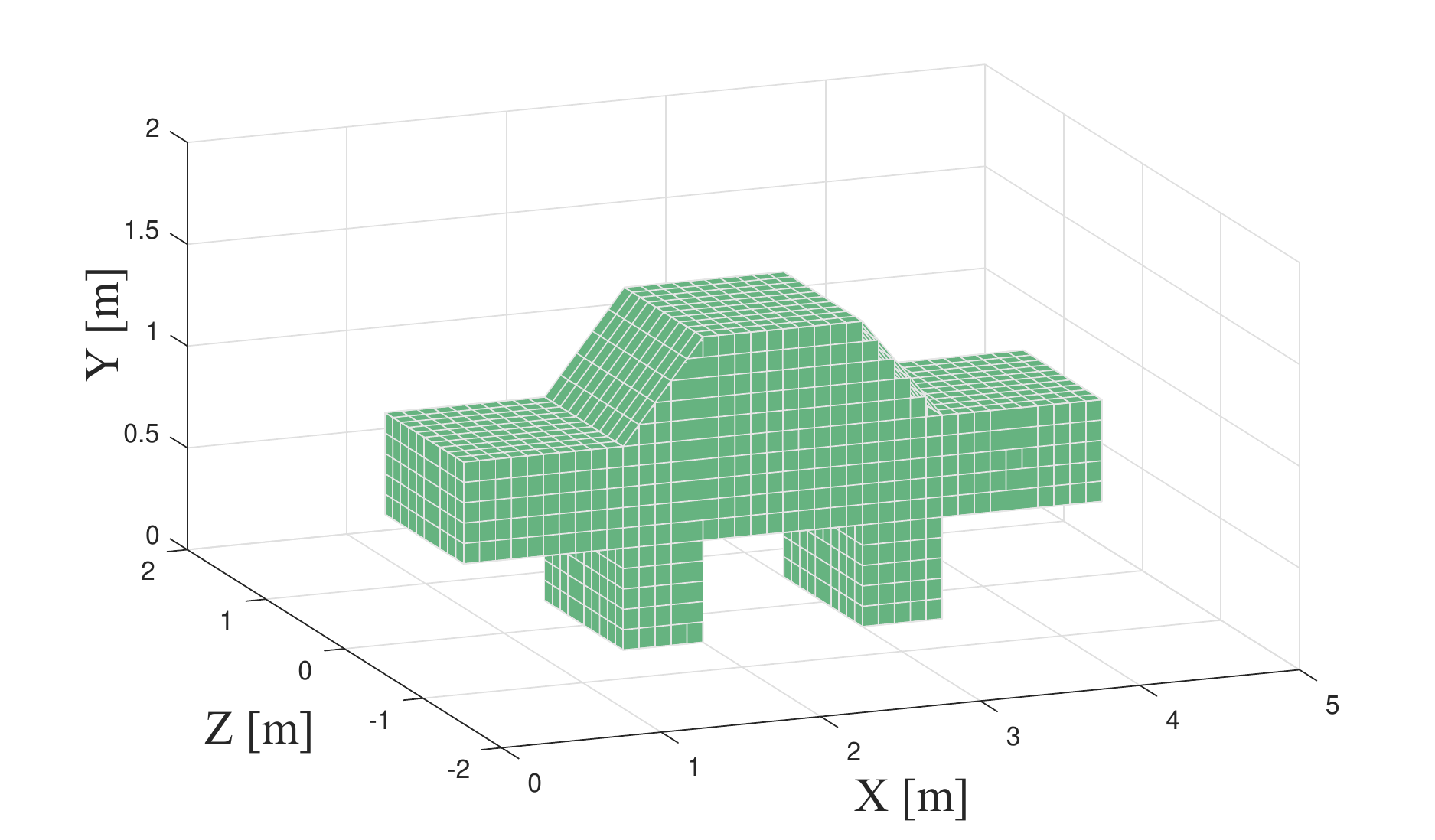}
		\caption*{(a) Envelop Diagram of TV.}%加*可以去掉默认前缀，作为图片单独的说明
		\label{fig:side:a}
	\end{minipage}
\hfill
	\begin{minipage}[t]{0.5\linewidth}%需要几张添加即可，注意设定合适的linewidth
		\centering
		\includegraphics[scale = 0.44]{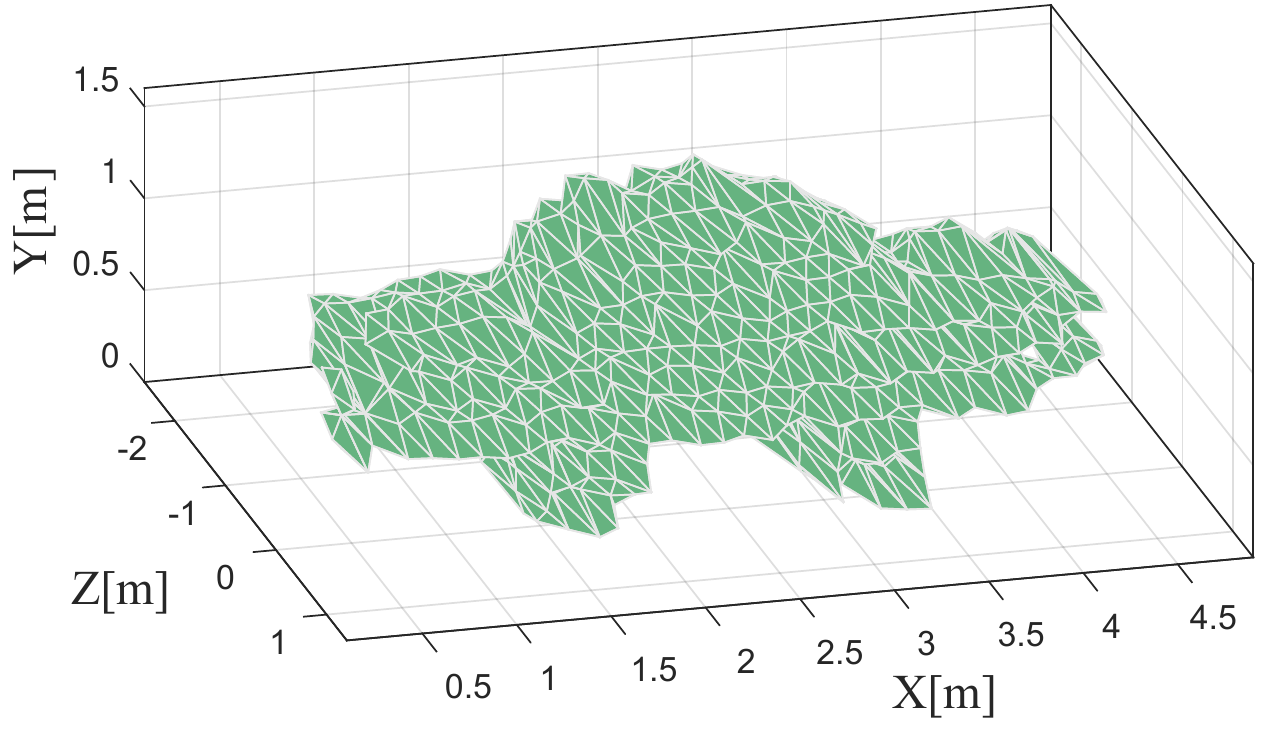}
		\caption*{(b) Envelop Diagram of the detected TV position.}
		\label{fig:side:b}
	\end{minipage}
	\caption{The envelop digaram comparison of the TV model and the detected position.}\label{EnvelopCompare}\vspace{-0.5cm}
\end{figure}
%	\begin{figure}
%		\centering
%		\includegraphics[scale = 0.6]{Figures/Hausdorff.pdf}
%		\caption{Hausdorff Distance of the recovered image versus TV-SV distance.}\label{Hausdorff}
%	\end{figure}
%	\begin{figure}
%		\begin{minipage}[t]{0.56\linewidth}%设定图片下字的宽度，在此基础尽量满足图片的长宽
%			\centering
%			\includegraphics[scale = 0.6]{Figures/DisErrorMirror.pdf}
%			\caption*{(a) Mean Location Error versus the TV-SV distance.}%加*可以去掉默认前缀，作为图片单独的说明
%			\label{Distance}
%		\end{minipage}
%		\begin{minipage}[t]{0.56\linewidth}%需要几张添加即可，注意设定合适的linewidth
%			\centering
%			\includegraphics[scale = 0.6]{Figures/NumMVError.pdf}
%			\caption*{(b) Mean Location Error versus the number of MVs where the distance is $8$m {\zzz(Old simulation results)}.}
%			\label{MVNumber}
%		\end{minipage}
%		\caption{The performance of the reconstructed images.}\label{PerformanceCompare}
%	\end{figure}
{
\subsection{Error in Clock Synchronization}
The clock difference detection error is omitted in procedures after clock synchronization according to the analysis in Proposition \ref{Prop:sync}. To justify such an assumption, we check it by simulations given in Fig. \ref{ClockSync}. It clearly shows that the clock difference detection error can be well suppressed by an appropriate number of antennas, e.g., $N_r = 16$ or $N_r = 64$, at the receiver. Therefore, it is reasonable to assume perfect clock synchronization for convenience.

\subsection{Different Interpolation Methods}
The performance by using different interpolation methods for resampling in Sec. \ref{sec:Imaging} is also checked by simulations in Fig. \ref{Interpolation}. It can be observed that all interpolation methods gives similar performance, where the errors are negligible for vehicular positioning. Therefore, we simply adopt linear interpolation in this work for low complexity.

\begin{figure}[h]
	\centering
	\includegraphics[%height=240pt, 
	width=450pt]{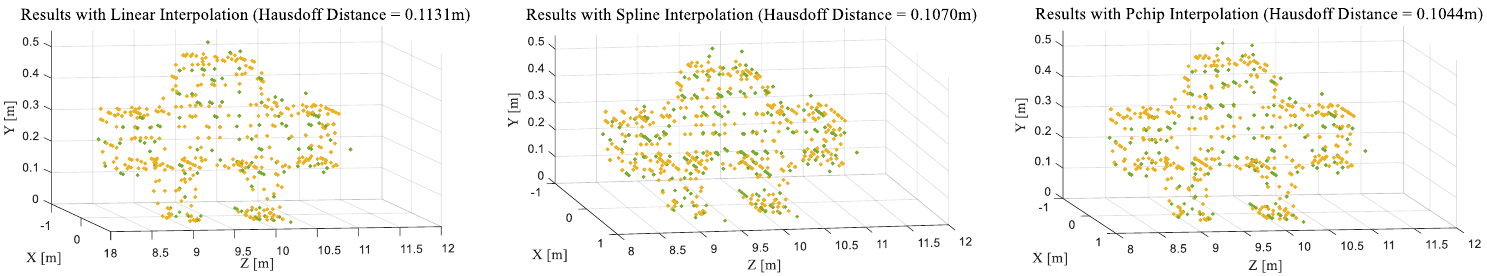}
	\caption{Comparison between different interpolation methods, where the initial TV model is in yellow, and the retrieved TV position by algorithm in Sec. \ref{sec:Imaging} is in green.}\label{Interpolation}\vspace{-0.5cm}
\end{figure}

\begin{figure}[t]
	\centering
	\includegraphics[%height=240pt, 
	width=200pt]{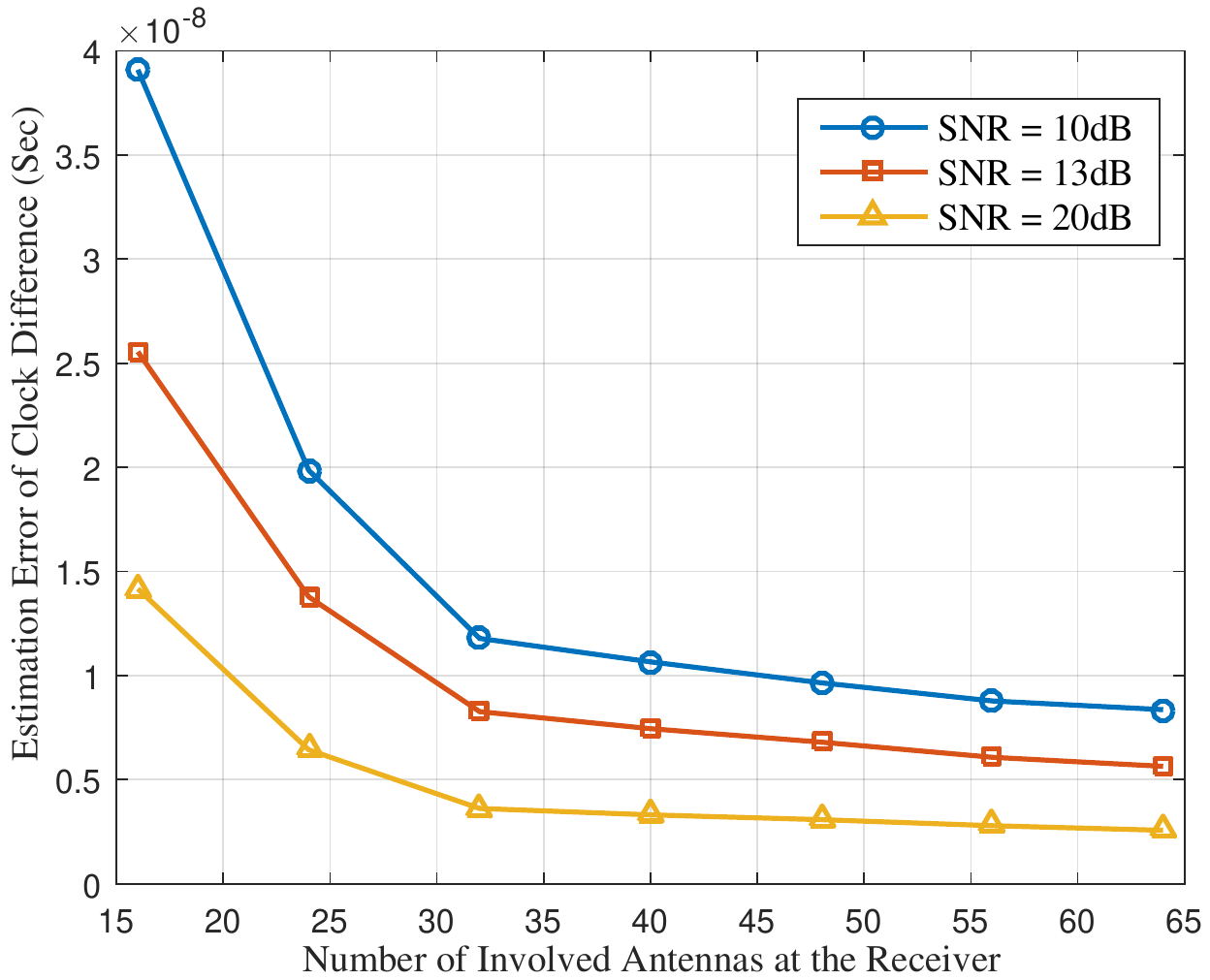}
	\caption{Estimation error of the clock difference versus the number of antennas involved at the receiver.}\label{ClockSync}\vspace{-0.5cm}
\end{figure} }

\subsection{Effect of Distance between SV and TV}
	In Fig.~\ref{Compare}, the Hausdorff distances are given under different TV-SV distances, showing that the positioning quality is degraded as the distance between the TV and SV increases.  This phenomenon can be explained by \eqref{spatialResolution}, where the spatial resolution becomes poor when the detection range $R$ is large. Therefore, the SV may not be able to capture the clear position of the TVs far away, especially when the number of reflection surfaces is small. Moreover, the relation between the Hausdorff distance and the TV-SV distance is not linear, which is led by the $\tan(\cdot)$ function in the solution \eqref{representXA}. With noise in consideration, the Hausdorff distance, indicating the comprehensive error level, increases faster when the TV-SV distance becomes larger.

\begin{figure}[t]
	\begin{minipage}[t]{0.5\linewidth}%设定图片下字的宽度，在此基础尽量满足图片的长宽
		\centering
		\includegraphics[scale = 0.45]{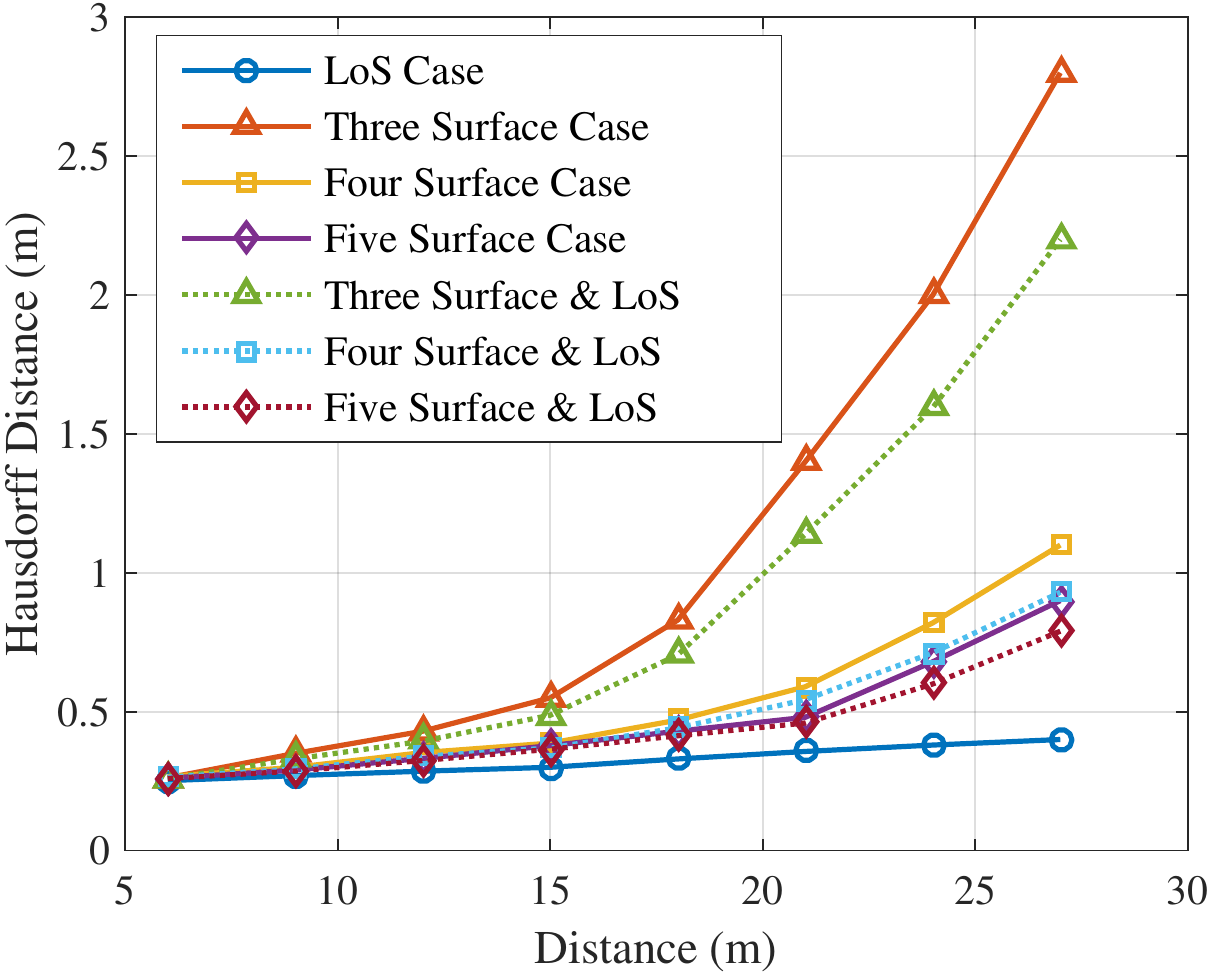}
		\caption*{(a) The  performance of the COMPOP with reflection links.}%加*可以去掉默认前缀，作为图片单独的说明
	\end{minipage}
	\hfill
	\begin{minipage}[t]{0.5\linewidth}%需要几张添加即可，注意设定合适的linewidth
		\centering
		\includegraphics[scale = 0.45]{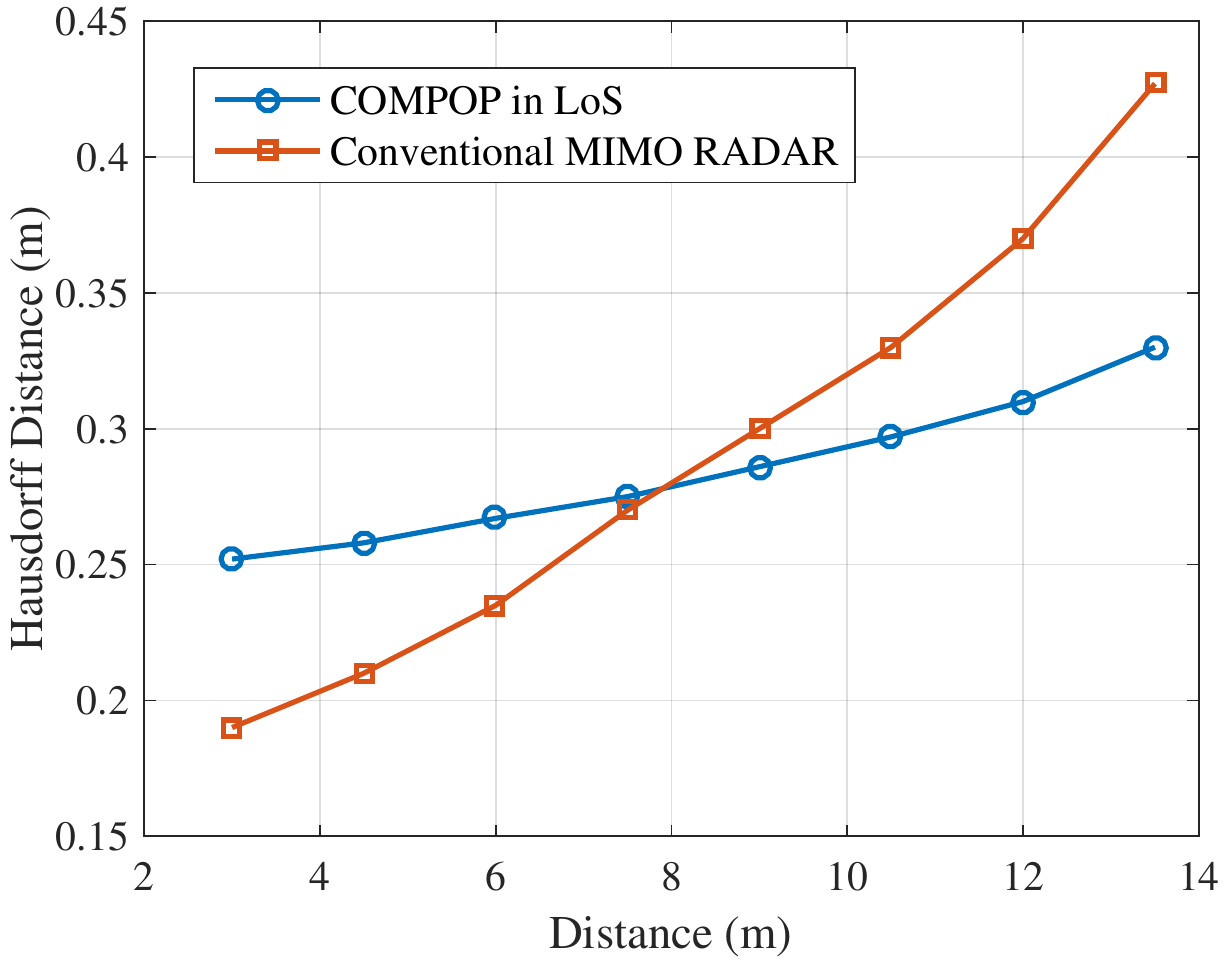}
		\caption*{(b) Comparison between the COMPOP and the existing MIMO RADAR.}
	\end{minipage}
	\caption{The performance of the proposed COMPOP versus distance.}\label{Compare}\vspace{-0.5cm}
\end{figure}

\subsection{Effect of Reflection Surface Number}
	The relation between the performance and the number of reflection surfaces is also presented in Fig.~\ref{Compare} (a). As explained in Sec. \ref{sec2:Common-point}, signals reflected from any three reflection surfaces give one estimation of $\theta_1$, but the resultant positioning quality is low due to the phase error, and the performance becomes unstable 
    when the TV-SV distance increases. We also consider the cases where more than $3$ reflection links exist $(L>3)$. Larger $L$ provides more combinations to estimate $\theta_1$, resulting in more accurate estimation of $\theta_1$ by canceling out individual estimated error. 
    Moreover, the LoS case leads to the minimum Hausdorff distance. It is considered as a lower bound because the operation in Sec. \ref{sec2:Common-point} is not involved, and the estimation error only comes from the approach in Sec. \ref{sec:Imaging}. The case with mixed LoS and NLoS paths are also investigated here, which is plotted by dotted lines. It can be observed that although the existence of LoS path enhances the performance, the errors from NLoS paths hamper the accuracy. 
 
{  
\subsection{Comparison to MIMO RADAR}
   The performance of conventional MIMO RADAR~\cite{BackPropagation,MIMO1} is also checked by simulations as shown in Fig.~\ref{Compare} (b), which is considered as a baseline of the proposed COMPOP design. The conventional MIMO RADAR is naturally synchronized with co-located transceivers and thus free from the clock difference detection error. Hence, conventional MIMO RADAR could be a better choice in short distance. However, when the distance becomes larger, the proposed COMPOP has a signal power gain and thus shows better performance as illustrated in Sec. \ref{Comparison_LoS}. }

%	Moreover, Fig. \ref{PerformanceCompare}(b) shows the relation between the number of MVs and the location information accuracy of the reconstructed image. Here another two optional MVs are prepared with plane equations
%	\begin{align}
%	S_4: z = -0.6x+8,\\
%	S_5: z = 1.8x+6.4.
%	\end{align}
%	From Fig. \ref{PerformanceCompare}(b), we can see that more accurate location information can be obtained when more MVs exists which offers propagation paths from the TV to the SV.

	\section{Conclusion}
	In this paper, a novel \emph{Coopertive Multi-point Positioning} (COMPOP) approach via mmWave signal transmissions has been proposed to capture the shape and location information of the TVs in both LoS and NLoS. The cooperative transmission between the TV and SV enables the real-time COMPOP without scanning process. The synchronization issue due to transceiver separation has been well addressed by a PDoA-based  positioning approach. In NLoS case, COMPOP establishes mirror-reflection links between the TV and SV under the assistance of the nearby vehicles. { The geometric relation between the virtual and actual TVs has been exploited to position the actual TV via an intelligent combining algorithm without priori knowledge on the nearby vehicles.}
	In conclusion, the proposed COMPOP is available in both LoS and NLoS situations with ultra-low latency and high detection accuracy, which is challenging for existing vehicular sensing techniques (e.g., RADAR and LIDAR). Therefore, this technique opens a new area of mmWave-based vehicular positioning and sensing. We believe the proposed COMPOP contributes to more intelligent and safer autonomous driving, and the potential of mmWave-based vehicular sensing can still be activated in the future.

	\appendix
{
	\subsection{Proof of proposition \ref{Prop:sync}}\label{A}
    Due to the assumption of Gaussian phase error, the covariance of the location estimation can be expressed as
    \begin{align}\label{Ap_Cov}
    {\rm{cov}}({\bf{\tilde x}}_{\mathsf a}) 
    =& \bf X {\sigma _z^2{{\bf{I}}_3}}, 
    \end{align}
    where $\bf X$ is a $3$-by-$3$ matrix as 
    \begin{align}\label{AP_EXP}
    \bf X&= {\left( {{\bf{G}}{{({\bf{\tilde x}}_{\mathsf {a}})}^T}{\bf{G}}({\bf{\tilde x}}_{\mathsf {a}})} \right)^{ - 1}}\nonumber\\
    &=
    \left(\sum_{m=2}^{N_r} \left[ {\begin{array}{*{20}{c}}
	{{{{\left( {\frac{{\partial {\mathsf{F_m}}({\bf{\tilde x}}_{\mathsf {a}})}}{{\partial x_{\mathsf {a}}}}} \right)}^2}} }&{ {\frac{{\partial {\mathsf{F_m}}({\bf{\tilde x}}_{\mathsf {a}})}}{{\partial x_{\mathsf {a}}}}\frac{{\partial {\mathsf{F_m}}({\bf{\tilde x}}_{\mathsf {a}})}}{{\partial y_{\mathsf {a}}}}} }&{ {\frac{{\partial {\mathsf{F_m}}({\bf{\tilde x}}_{\mathsf {a}})}}{{\partial x_{\mathsf {a}}}}\frac{{\partial {\mathsf{F_m}}({\bf{\tilde x}}_{\mathsf {a}})}}{{\partial z_{\mathsf {a}}}}} }\\
	{ {\frac{{\partial {\mathsf{F_m}}({\bf{\tilde x}}_{\mathsf {a}})}}{{\partial x_{\mathsf {a}}}}\frac{{\partial {\mathsf{F_m}}({\bf{\tilde x}}_{\mathsf {a}})}}{{\partial y_{\mathsf {a}}}}} }&{ {{{\left( {\frac{{\partial {\mathsf{F_m}}({\bf{\tilde x}}_{\mathsf {a}})}}{{\partial y_{\mathsf {a}}}}} \right)}^2}} }&{ {\frac{{\partial {\mathsf{F_m}}({\bf{\tilde x}}_{\mathsf {a}})}}{{\partial y_{\mathsf {a}}}}\frac{{\partial {\mathsf{F_m}}({\bf{\tilde x}}_{\mathsf {a}})}}{{\partial z_{\mathsf {a}}}}} }\\
	{ {\frac{{\partial {\mathsf{F_m}}({\bf{\tilde x}}_{\mathsf {a}})}}{{\partial x_{\mathsf {a}}}}\frac{{\partial {\mathsf{F_m}}({\bf{\tilde x}}_{\mathsf {a}})}}{{\partial z_{\mathsf {a}}}}} }&{ {\frac{{\partial {\mathsf{F_m}}({\bf{\tilde x}}_{\mathsf {a}})}}{{\partial y_{\mathsf {a}}}}\frac{{\partial {\mathsf{F_m}}({\bf{\tilde x}}_{\mathsf {a}})}}{{\partial z_{\mathsf {a}}}}} }&{ {{{\left( {\frac{{\partial {\mathsf{F_m}}({\bf{\tilde x}}_{\mathsf {a}})}}{{\partial z_{\mathsf {a}}}}} \right)}^2}} }
	\end{array}} \right]\right)^{ - 1}\nonumber\\
    &=\frac{1}{{N_r}-1}\left(\frac{1}{{N_r}-1}\sum_{m=2}^{N_r} \left[ {\begin{array}{*{20}{c}}
	{{{{\left( {\frac{{\partial {\mathsf{F_m}}({\bf{\tilde x}}_{\mathsf {a}})}}{{\partial x_{\mathsf {a}}}}} \right)}^2}} }&{ {\frac{{\partial {\mathsf{F_m}}({\bf{\tilde x}}_{\mathsf {a}})}}{{\partial x_{\mathsf {a}}}}\frac{{\partial {\mathsf{F_m}}({\bf{\tilde x}}_{\mathsf {a}})}}{{\partial y_{\mathsf {a}}}}} }&{ {\frac{{\partial {\mathsf{F_m}}({\bf{\tilde x}}_{\mathsf {a}})}}{{\partial x_{\mathsf {a}}}}\frac{{\partial {\mathsf{F_m}}({\bf{\tilde x}}_{\mathsf {a}})}}{{\partial z_{\mathsf {a}}}}} }\\
	{ {\frac{{\partial {\mathsf{F_m}}({\bf{\tilde x}}_{\mathsf {a}})}}{{\partial x_{\mathsf {a}}}}\frac{{\partial {\mathsf{F_m}}({\bf{\tilde x}}_{\mathsf {a}})}}{{\partial y_{\mathsf {a}}}}} }&{ {{{\left( {\frac{{\partial {\mathsf{F_m}}({\bf{\tilde x}}_{\mathsf {a}})}}{{\partial y_{\mathsf {a}}}}} \right)}^2}} }&{ {\frac{{\partial {\mathsf{F_m}}({\bf{\tilde x}}_{\mathsf {a}})}}{{\partial y_{\mathsf {a}}}}\frac{{\partial {\mathsf{F_m}}({\bf{\tilde x}}_{\mathsf {a}})}}{{\partial z_{\mathsf {a}}}}} }\\
	{ {\frac{{\partial {\mathsf{F_m}}({\bf{\tilde x}}_{\mathsf {a}})}}{{\partial x_{\mathsf {a}}}}\frac{{\partial {\mathsf{F_m}}({\bf{\tilde x}}_{\mathsf {a}})}}{{\partial z_{\mathsf {a}}}}} }&{ {\frac{{\partial {\mathsf{F_m}}({\bf{\tilde x}}_{\mathsf {a}})}}{{\partial y_{\mathsf {a}}}}\frac{{\partial {\mathsf{F_m}}({\bf{\tilde x}}_{\mathsf {a}})}}{{\partial z_{\mathsf {a}}}}} }&{ {{{\left( {\frac{{\partial {\mathsf{F_m}}({\bf{\tilde x}}_{\mathsf {a}})}}{{\partial z_{\mathsf {a}}}}} \right)}^2}} }
	\end{array}} \right]\right)^{ - 1},
    \end{align}
    where the inverse matrix's each component converges to its expectation as $N_r$ becomes large, which becomes independent to $N_r$. In other words, $\bf X$ is proportional to $\frac{1}{N_r-1}$. }

	\subsection{Proof of Lemma \ref{lemma1}}\label{B}
	We adopt the scalar diffraction idea in~\cite{AngularSpectrum} that the wave field can be decomposed as an infinite integral of planar waves, given as
	\begin{align}\label{decomposition}
	{\exp\!\left(\! - j\frac{2\pi f}{c}D{(\mathbf{x}, \mathbf{p}_{m})} \!\right)} \!=\! \iint_{f^{(z)} = \sqrt{{{f}^2} - \left(f^{(x)}\right)^2 - \left(f^{(y)}\right)^2}}\exp\left(\!-j\frac{2\pi}{c} \mathbf{f}^T (\mathbf{x}-\mathbf{p}_m) \! \right) \!d{{f}^{(x)}}\!d{{f}^{(y)}},
	\end{align}	
	where $\mathbf{f}\!=\!\left[f^{(x)}, f^{(y)}, f^{(z)}\right]^T$ denotes a spatial frequency vector. %and the frequency $f$ corresponds to the spatial frequency vectors $\left\{\mathbf{f} \ \bigg| \ \sqrt {\left(f^{(x)}\right)^2 + \left(f^{(y)}\right)^2 + \left(f^{(z)}\right)^2} = f \right\}$.
	Accordingly, the signal in \eqref{imageSignal} can be expanded by rewriting the exponential term in terms of $\mathbf{f}$ as 
	\begin{align}
	y(\mathbf{p}_m, f)  &= \int_{\mathbb{R}^3}  {\mathsf{I}_{\{\mathbf{x} \in \mathcal{X}\}}\left\{\iint_{f^{(z)} = \sqrt{{{f}^2} - \left(f^{(x)}\right)^2 - \left(f^{(y)}\right)^2}}\exp\left(-j\frac{2\pi}{c} \mathbf{f}^T (\mathbf{x}-\mathbf{p}_m) \right) d{{f}^{(x)}}\!d{{f}^{(y)}} \right\}d{\mathbf{x}}}\nonumber\\
	&=\!\!\iint_{f^{(z)} = \sqrt{{{f}^2} - \left(f^{(x)}\right)^2 - \left(f^{(y)}\right)^2}} 
	\left\{\int_{\mathbb{R}^3}  \mathsf{I}_{\{\mathbf{x} \in \mathcal{X}\}} \exp\left(-j\frac{2\pi}{c} \mathbf{f}^T (\mathbf{x}-\mathbf{p}_m) \right)d{\mathbf{x}}\right\} 
	d{{f}^{(x)}}\!d{{f}^{(y)}}\nonumber\\[0.1mm]
	&=\!\!\iint_{\!f^{(z)} = \sqrt{\!\!{{f}^2} - \left(\!f^{(x)}\!\right)^2 - \left(\!f^{(y)}\!\right)^2}}\! 
	\underbrace {  \left\{\!\int_{\mathbb{R}^3}\!\! \mathsf{I}_{\{\mathbf{x} \in \mathcal{X}\}}\! \exp\!\left(\!\!-j\frac{2\pi}{c} \mathbf{f}^T \!{\mathbf{x}}
		\! \right)\!d{\mathbf{x}}\!\right\}}_{\mathsf{FT}_{\mathsf{3D}} (\mathsf{I}_{\{\mathbf{x} \in \mathcal{X}\}})} 
	\!\exp\!\left(\!j\frac{2\pi}{c}\mathbf{f}^T\!\mathbf{p}_m\!\right)\!d{{f}^{(x)}}\!d{{f}^{(y)}},\nonumber
	\end{align}
	%where $\mathsf{FT}_{\mathsf{3D}}(\cdot)$ represents the 3D Fourier transform defined in Table \ref{TableFourier}. 
	Then consider the received signal at the 2D plane $z=0$ where $\mathbf{p}_m = (p_m^{(x)},p_m^{(y)},0)^T$, given as
	\begin{align}\label{finalExpression}
	{y}\left(p_m^{(x)},p_m^{(y)},0,f\right)\!&=\!\!\iint\bigg\{\mathsf{FT}_{\mathsf{3D}}  (\mathsf{I}_{\{\mathbf{x} \in \mathcal{X}\}}\!)\big|_{f^{(z)} = \sqrt{\!\!{{f}^2} - \left(\!f^{(x)}\!\right)^2 - \left(\!f^{(y)}\!\right)^2}}\bigg\} \exp\!\left(\!j\frac{2\pi}{c}\mathbf{f}^T\!\mathbf{p}_m\!\right)\!d{{f}^{(x)}}\!d{{f}^{(y)}}\nonumber\\
	&=\mathsf{FT}_{\mathsf{2D}}^{-1}\bigg\{\mathsf{FT}_{\mathsf{3D}} (\mathsf{I}_{\{\mathbf{x} \in \mathcal{X}\}})\big|_{f^{(z)} = \sqrt{\!\!{{f}^2} - \left(\!f^{(x)}\!\right)^2 - \left(\!f^{(y)}\!\right)^2}} \bigg\}. 
	\end{align}
	This finishes the proof.
	%where $\mathsf{FT}_{\mathsf{2D}}^{-1}(\cdot)$ is the 2D inverse Fourier transform.
	%,  and $\mathsf{Q}_k(\mathbf{f})$ is an indicator function given as
	%	\begin{align}
	%\mathsf{Q}_k(\mathbf{f}) = \left\{ \begin{array}{l}
	%1, \qquad \|\mathbf{f}\| = {f},\ f^{(z)}\ge 0,\\
	%	0, \qquad \text{otherwise}
	%	\end{array} \right.
	%	\end{align}
	%Recall that $Z$-axis is in the direction of AoA known at the SV. Therefore, the results in \eqref{finalExpression} can be further approximated as
	%\begin{align}\label{approx}
	%{y}\left(\mathbf{p}_m, f\right)\exp\left(\!-j\frac{2\pi f}{c}p_m^{(z)}\!\right) %&\approx {y}\!\left(p_m^{(x)},p_m^{(y)},0, f_k\right)\nonumber\\
	%&=\mathsf{FT}_{\mathsf{2D}}^{-1}\bigg\{\!\mathsf{FT}_{\mathsf{3D}} (\mathsf{I}_{\{\mathbf{x} \in \mathcal{X}\}})\big|_{f^{(z)} = \sqrt{\!\!{{f}^2} - \left(\!f^{(x)}\!\right)^2 - \left(\!f^{(y)}\!\right)^2}}\bigg\},
	%\end{align}
	%which holds tightly when the the TV-SV distance is much larger than the SV's size.
	
	\subsection{Proof of Lemma \ref{CommonPointRelation} } \label{C}
	In Fig.~\ref{location}, we have $\theta_{\ell_1}  \in \left[ { - {\pi },{\pi }} \right]$ and the equations of line $l_{{\bf x}_{\mathsf{a}}^{(\ell_1)}{\bf x}_{\mathsf{a}}}$ and $l_{{\bf x}_{\mathsf{a}}^{(\ell_2)}{\bf x}_{\mathsf{a}}}$ are
	\begin{align}
	l_{{\bf x}_{\mathsf{a}}^{(\ell_1)}{\bf x}_{\mathsf{a}}}:\quad{\rm{  }}z = {z_{\mathsf{a}}^{(\ell_1)}} + \tan ({\theta_{\ell_1}})\left( {x - {x_{\mathsf{a}}^{(\ell_1)}}} \right);
	\qquad l_{{\bf x}_{\mathsf{a}}^{(\ell_2)}{\bf x}_{\mathsf{a}}}:\quad{\rm{  }}z = {z_{\mathsf{a}}^{(\ell_2)}} + \tan ({\theta_{\ell_2}})\left( {x - {x_{\mathsf{a}}^{(\ell_2)}}} \right).\nonumber
	\end{align}
	Thus the intersection point can be solved as 
	\begin{align}\label{ApprepresentXA}
	\left\{ \begin{array}{l}
	{x_{\mathsf{a}}} = \frac{{\left( {z_{\mathsf{a}}^{({\ell _1})} - z_{\mathsf{a}}^{({\ell _2})}} \right) + \left( {x_{\mathsf{a}}^{({\ell _2})}\tan ({\theta _{{\ell _2}}}) - x_{\mathsf{a}}^{({\ell _1})}\tan ({\theta _{{\ell _1}}})} \right)}}{{\tan ({\theta _{{\ell _2}}}) - \tan ({\theta _{{\ell _1}}})}}\\
	{z_{\mathsf{a}}} = z_{\mathsf{a}}^{({\ell _1})} + \tan ({\theta _{{\ell _1}}})\big( {{x_{\mathsf{a}}} - x_{\mathsf{a}}^{({\ell _1})}} \big)
	\end{array} \right..
	\end{align}
	The general relations in \eqref{representXA} can be derived similarly.
	Due to the symmetric relation between the virtual and actual TVs,  we have
	\begin{align}\label{AppAngRelation}
	\pi-\varphi = {\varphi_{\ell_1}} - 2{\theta_{\ell_1}} = {\varphi_{\ell_2}} - 2{\theta_{\ell_2}},
	\end{align}
	%which can be observed from Fig.~\ref{location}. Thus \eqref{angleRelation} is obtained. 
	%\begin{figure}[t]\vspace{-0.5cm}
		%\centering
		%\includegraphics[scale = 0.5]{Figures/angleRelation.pdf}
		%\caption{The geometry relations among different virtual TVs and the actual TV.}\label{AngRelation}\vspace{-0.5cm}
	%\end{figure}
	
	\subsection{Proof of Proposition \ref{Prop2} } \label{D}
	Bring \eqref{AppAngRelation} into (\ref{ApprepresentXA}), $x_{\mathsf{a}}$ can be simplified as
	\begin{align}\label{Applocation_X1_12}
	x_{\mathsf{a}} = \frac{{\big( {{z_{\mathsf{a}}^{(\ell_1)}} - {z_{\mathsf{a}}^{(\ell_2)}}} \big) + {x_{\mathsf{a}}^{(\ell_2)}}\tan \big( {{\theta_{\ell_1}} - \frac{{{\varphi_{\ell_1}} - {\varphi_{\ell_2}}}}{2}} \big) - {x_{\mathsf{a}}^{(\ell_1)}}\tan ({\theta_{\ell_1}})}}{{\tan \big( {{\theta_{\ell_1}} - \frac{{{\varphi_{\ell_1}} - {\varphi_{\ell_2}}}}{2}} \big) - \tan ({\theta_{\ell_1}})}},
	\end{align}
	which indicates that $x_{\mathsf{a}}$ is only determined by the angle $\theta_{\ell_1}  \in \left[ { - {\pi },{\pi }} \right]$. Given the angle $\theta_{\ell_1} $, the estimation of $x_{\mathsf{a}}$ from (\ref{Applocation_X1_12}) is denoted as ${\hat x_{\mathsf{a}}}$.	
	Similarly, another estimation of $x_1$ can be obtained from the common points $({\bf x}_{\mathsf{a}}^{(\ell_1)},{\bf x}_{\mathsf{b}}^{(\ell_1)})$ and $({{\bf x}_{\mathsf{a}}^{(\ell_3)},{\bf x}_{\mathsf{b}}^{(\ell_3)}})$, which is given as
	\begin{align}\label{Applocation_X1_13}
	x_{\mathsf{a}} = \frac{{\big( {{z_{\mathsf{a}}^{(\ell_1)}} - {z_{\mathsf{a}}^{(\ell_3)}}} \big) + {x_{\mathsf{a}}^{(\ell_3)}}\tan \big( {{\theta_{\ell_1}} - \frac{{{\varphi_{\ell_1}} - {\varphi_{\ell_3}}}}{2}} \big) - {x_{\mathsf{a}}^{(\ell_1)}}\tan ({\theta_{\ell_1}})}}{{\tan \big( {{\theta_{\ell_1}} - \frac{{{\varphi_{\ell_1}} - {\varphi_{\ell_3}}}}{2}} \big) - \tan ({\theta_{\ell_1}})}}.
	\end{align}
	The estimation of $x_{\mathsf{a}}$ from (\ref{Applocation_X1_13}) is denoted as ${\tilde x_{\mathsf{a}}}$, and it can be observed from (\ref{Applocation_X1_13}) that ${\tilde x_{\mathsf{a}}}$ is also only determined by $\theta_{\ell_1}$. 
	Therefore, the SV can search $\theta_{\ell_1}$ in the range of $\left[ { - {\pi },{\pi }} \right]$ to minimize $\left| {\hat x_{\mathsf{a}}}-{\tilde x_{\mathsf{a}}}\right|$. Two solutions can be obtained with the optimal $\theta_{\ell_1}$, denoted as ${{\hat x}_{\mathsf{a}}^*}$ and ${{\tilde x}_{\mathsf{a}}^*}$. Then the optimal solution of $x_{\mathsf{a}}$ is given as $x_{\mathsf{a}}^* = \frac{{{\hat x}_{\mathsf{a}}^*}+{{\tilde x}_{\mathsf{a}}^*}}{2}$.  With $x_{\mathsf{a}}^*$,  the location ${\bf x}_{\mathsf{a}}$ can be calculated according to \eqref{representXA}, so as ${\bf x}_{\mathsf{b}}$. Therefore, three virtual TVs is enough for the SV to detect the position of the actual TV. This finishes the proof.

	\subsection{Proof of Proposition \ref{Prop3} } \label{E}
     Based on the line function \eqref{reflectionPlane}, as well as the symmetric geometry relation between virtual TV $\ell$ and the actual one, it is easy to establish the mathematical relation between 
     $\mathbf{x}^{(\ell)}$ and $\mathbf{x}$ as
     \begin{align}\label{locRule}
	    \left\{ \begin{array}{l}
		x = {x^{(\ell )}} + \Delta {x^{(\ell )}}\\
		z = {z^{(\ell )}} + \tan ({\theta _\ell^* }) \cdot \Delta {x^{(\ell )}}
	    \end{array} \right.,
     \end{align}
     where ${\Delta {x^{(\ell )}}}$ is the $x$-direction projection of distance between $(x,z)$ and $(x^{(\ell)},z^{(\ell)})$. The middle point of $(x,z)$ and $(x^{(\ell)},z^{(\ell)})$ locates at the surface $\ell$, meaning that the middle point $\big( {{x^{(\ell )}} + \frac{{\Delta {x^{(\ell )}}}}{2},{z^{(\ell )}} + \tan ({\theta _\ell^* }) \cdot \frac{{\Delta {x^{(\ell )}}}}{2}} \big)$ should satisfy function \eqref{reflectionPlane}. Thus we have
     \begin{align}\label{Deltax}
     	z_{\mathsf{a}}^{(\ell )} + {z_{\mathsf{a}}} - \frac{1}{{\tan ({\theta _\ell^* })}}\left( 2{{x^{(\ell )}} + \Delta {x^{(\ell )}} - x_{\mathsf{a}}^{(\ell )} - {x_{\mathsf{a}}}} \right) = 2{z^{(\ell )}} + \tan ({\theta _\ell^* }) \cdot \Delta {x^{(\ell )}}.
     \end{align}
     The result derived from \eqref{Deltax} is given as
     \begin{align}\label{DeriveDeltax}
     	{\Delta {x^{(\ell )}}}=\frac{{\tan (\theta _\ell ^*)}}{1+{\tan^2 (\theta _\ell ^*)} }\bigg( {\frac{{x_{\mathsf{a}}^{(\ell )} + x_{\mathsf{a}}^*}}{{\tan (\theta _\ell ^*)}} + z_{\mathsf{a}}^{(\ell )} + z_{\mathsf{a}}^* - \frac{{2{x^{(\ell )}}}}{{\tan (\theta _\ell ^*)}} - 2{z^{(\ell )}}} \bigg).
     \end{align}
     Bring \eqref{DeriveDeltax} into \eqref{locRule}, the result \eqref{FindLocation} can be obtained by replacing $(x,z)$ with  $(x^*,z^*)$.

	\bibliographystyle{ieeetran}%
	\bibliography{mirror}

	% that's all folks
\end{document}